\include{birkart.sty}

\documentclass{amsart}[12pt]
\usepackage{amsmath}
\usepackage{amssymb}
\usepackage{mathrsfs}
\usepackage{amsfonts}
\usepackage{graphicx}
\usepackage{tikz}
\usepackage{texdraw}
\usepackage{graphpap}
\usepackage{enumitem}
\usepackage{chngcntr}
\usepackage{wasysym}
\usepackage{color}
\usepackage[OT2,T1]{fontenc}
\usepackage{ulem}

\usetikzlibrary{decorations.markings}

\input txdtools

\theoremstyle{plain}
\newtheorem{mth}{Theorem}[section]

\newtheorem{lem}[mth]{Lemma}

\theoremstyle{remark}
\newtheorem*{defn}{Definition}
\newtheorem*{rem}{Remark}

\newtheorem*{ex}{Example}

\numberwithin{equation}{section}

\counterwithout{figure}{section}

\newcommand{\I}{\mathrm{I}}
\newcommand{\II}{\mathrm{II}}
\newcommand{\III}{\mathrm{III}}

\newcommand{\pert}{u}

\newcommand{\bfR}{\mathbf{R}}

\newcommand{\point}{a}

\newcommand{\calB}{{\mathcal B}}

\newcommand{\calH}{{\mathscr H}}

\newcommand{\calN}{{\mathscr N}}

\newcommand{\calQ}{{\mathcal Q}}

\newcommand{\calP}{{\mathscr P}}

\newcommand{\calV}{{\mathscr V}}

\newcommand{\calR}{{\mathcal R}}

\newcommand{\bfP}{{\mathbf P}}
\newcommand{\bfK}{{\mathbf K}}
\newcommand{\bfE}{{\mathbf E}}

\newcommand{\bfL}{{\mathbf L}}
\newcommand{\bfk}{{\mathbf k}}

\newcommand{\Int}{\operatorname{Int}}

\newcommand{\trace}{\operatorname{trace}}

\newcommand{\tr}{\operatorname{trace}}
\newcommand{\fl}{\operatorname{fluct}}
\newcommand{\R}{{\mathbb R}}

\newcommand{\D}{D}
\newcommand{\Z}{{\mathbb Z}}

\newcommand{\C}{{\mathbb C}}

\newcommand{\const}{\mathrm{const.}}
\newcommand{\eps}{{\varepsilon}}

\newcommand{\re}{\operatorname{Re}}
\newcommand{\im}{\operatorname{Im}}

\newcommand{\Ham}{\mathbf{H}}
\newcommand{\calZ}{\mathscr{Z}}
\newcommand{\Pol}{\operatorname{Pol}}

\renewcommand{\d}{{\partial}}
\newcommand{\dbar}{\bar{\partial}}
\newcommand{\1}{\mathbf{1}}
\newcommand{\Pc}{\operatorname{Pc}}

\newcommand{\Var}{\operatorname{Var}}

\newcommand{\dist}{\operatorname{dist}}
\newcommand{\supp}{\operatorname{supp}}

\newcommand{\Lap}{\Delta}
\newcommand{\lnorm}{\left\|}
\newcommand{\rnorm}{\right\|}

\def\norm#1{\lnorm {#1} \rnorm}

\def\labs{\left |}
\def\rabs{\right |}
\def\babs#1{\labs {#1} \rabs}

\makeatletter
\newcommand*\bigcdot{\mathpalette\bigcdot@{.5}}
\newcommand*\bigcdot@[2]{\mathbin{\vcenter{\hbox{\scalebox{#2}{$\m@th#1\bullet$}}}}}
\makeatother

\begin{document}

\title[Insertion of a point charge]{The random normal matrix model: insertion of a point charge}

\author{Yacin Ameur}

\address{Yacin Ameur\\
Department of Mathematics\\
Faculty of Science\\
Lund University\\
P.O. BOX 118\\
221 00 Lund\\
Sweden}
\email{Yacin.Ameur@maths.lth.se}

\author{Nam-Gyu Kang}
\address{Nam-Gyu Kang\\
School of Mathematics\\
Korea Institute for Advanced Study\\
85 Hoegiro\\
Dongdaemun-gu\\
Seoul 02455\\
Republic of Korea}
\email{namgyu@kias.re.kr}

\author{Seong-Mi Seo}
\address{Seong-Mi Seo\\
School of Mathematics\\
Korea Institute for Advanced Study\\
85 Hoegiro\\
Dongdaemun-gu\\
Seoul 02455\\
Republic of Korea}
\email{seongmi@kias.re.kr}


\thanks{Nam-Gyu Kang was partially supported by Samsung Science and Technology Foundation (SSTF-BA1401-51), a
KIAS Individual Grant (MG058103) at Korea Institute for Advanced Study, and National Research Foundation of
Korea under grant number NRF-2019R1A5A1028324.}

\thanks{Seong-Mi Seo was partially supported by a KIAS Individual Grant (MG063103) at Korea Institute for Advanced Study and by the National Research Foundation of Korea, grant number 2019R1F1A1058006 and NRF-2019R1A5A1028324.}

\begin{abstract} In this article, we study microscopic properties of a two-dimensional Coulomb gas ensemble
near a conical singularity arising from insertion of a point charge in the bulk of the droplet. In the determinantal case,
we characterize all rotationally symmetric
scaling limits (``Mittag-Leffler fields'') and obtain universality of them when the underlying potential is algebraic.
Applications include a central limit theorem for $\log|p_n(\zeta)|$ where $p_n$ is the characteristic polynomial of an $n$:th order random normal matrix.
\end{abstract}

\maketitle

\section{Introduction and main results} \label{imr}
Insertion operations appear frequently in Coulomb gas theory, conformal field theory, and allied topics. In this article we study planar Coulomb gas ensembles $\{\zeta_j\}_1^n$ in the determinantal case, under the influence of an external confining potential $Q$ and we investigate the effect of inserting a point charge at a point $p$ in the bulk of the droplet created by the external field. The inserted charge is assumed to be of the same order of magnitude as the individual charges $\zeta_j$.

 From a field-theoretical point of view, this kind of insertion operation is quite natural, for example it can be directly related with
the statistics of the random logarithmic potential $Y_n(p)=\sum_1^n\log|p-\zeta_j|$. The investigation of random fields such as $Y_n$ and their asymptotic properties as $n\to\infty$
is currently an active area; see for instance \cite{DS,DS2,FHK,FKS,HKO01,KM,KM0,Ka,Kra,La18,La20,LP19,LCCW,LY,S2,WW} for related works.
Our main contribution in this direction is a central limit theorem for $Y_n(p)$ which is valid for the class of
\textit{algebraic} potentials (of the form $Q(\zeta)=\tau_0|\zeta-p|^{2\lambda}+P(\zeta)$ near the droplet, where $\tau_0,\lambda>0$, $P$ is a harmonic polynomial, and (say) $Q\equiv+\infty$ near $\infty$).

 Our exposition depends on a knowledge of the microscopic scaling limit near the inserted charge. Assuming that the potential is algebraic and using theory for rescaled Ward identities in the spirit of the papers \cite{AHM2,AK,AKM,AS}, we are in fact able to find an explicit microscopic correlation kernel, depending on three parameters: $\tau_0$, $\lambda$, and the strength of the inserted charge.
This part of our work is connected with the theory of certain Fock-Sobolev spaces of entire functions, and might be of some independent interest.

\subsubsection*{General notation} \label{notcon} If $g$ is a function, then $\bar{g}(z)$ means the complex-conjugate of $g(z)$.
A function $h(\zeta,\eta)$ is \textit{Hermitian} if $h(\eta,\zeta)=\bar{h}(\zeta,\eta)$ and \textit{Hermitian-analytic} (\textit{entire}) if
it is moreover analytic (entire) in $\zeta$ and $\bar{\eta}$.
A \textit{cocycle} is a function
$c(\zeta,\eta)=g(\zeta)\bar{g}(\eta)$ where $g$ is a continuous unimodular function.
$\C^*=\C\setminus\{0\}$ is the punctured plane, $\hat{\C}=\C\cup\{\infty\}$ is the Riemann sphere, $\D(a;r)$ is the open disk with center $a$, radius $r$, and $\Pol(n)$ is the set of analytic polynomials of degree at most $n-1$. If $K\subset\C$ is compact then the ``polynomially convex hull'' $\Pc K$ is the complement of the unbounded component of $\C\setminus K$. (Cf. e.g. \cite[p. 53]{ST}.)

We write $\Lap=\d\dbar=\tfrac 1 4(\tfrac {\d^2} {\d x^2}+\tfrac {\d^2} {\d y^2})$ (usual Laplacian
divided by $4$) and $dA=\tfrac 1 \pi dxdy$ (area measure divided by $\pi$).
Also, we sometimes write $\mu(f)$ for $\int f\, d\mu$.

\subsection{The Random Normal Matrix model} \label{bset} We now introduce our basic setup. In what follows, we will assume some basic familiarity with reproducing kernels and Bergman spaces; see \cite{DuS} and \cite{Ar} for introductions
to these theories.

We start by fixing a suitable lower semicontinuous function (external potential) $Q:\C\to \R\cup\{+\infty\}$ of sufficient increase near $\infty$,
\begin{equation}\label{gro}\liminf_{\zeta\to\infty}\frac {Q(\zeta)}{\log|\zeta|}>2
\end{equation}
For a Borel measure $\mu$ on $\C$ we
associate the weighted logarithmic energy,
$$I_{Q}[\mu]=\iint _{\C^2} \log \frac{1}{|\zeta-\eta|}\, d\mu(\zeta)d\mu(\eta)+\mu(Q).$$
The minimizer $\sigma=\sigma[Q]$ of this energy, among compactly supported Borel probability measures $\mu$, is
the equilibrium measure in external potential $Q$. The
support $S=S[Q]:=\supp\sigma$ is called the \textit{droplet}. Assuming
some smoothness of $Q$ in a neighbourhood of the droplet we ensure that $S$ is compact and  $\sigma$ is a probability measure taking the form
$d\sigma=\Lap Q\cdot \1_S\, dA.$
In particular we have $\Lap Q\ge 0$ on $S$. (We refer to~\cite{ST} as a source
for these results; see especially Theorem I.1.3, Theorem II.1.3. In our case $Q$ will not be everywhere smooth near the droplet, but merely satisfy a H\"{o}lder condition at the origin,
so the remarks about nonsmooth potentials on page 138 are also relevant.)

In the following, we will need some further regularity of the droplet.
To guarantee this, by means of Sakai's theory, it suffices, besides the growth \eqref{gro}, to assume that $Q$ be real-analytic in some neighbourhood of the boundary $\d S$. We refer to
the discussion in \cite[Sect. 1]{AKMW} for further details.

To simplify the discussion, we will assume that $Q$ be real-analytic wherever $Q<+\infty$, except possibly at the origin, where we merely assume that the limit
\begin{equation}
\label{lam_lim}\lambda-1:=\lim_{\zeta\to 0}\frac {\log\Lap Q(\zeta)}{2\log|\zeta|}
\end{equation}
exists and exceeds $-1$, i.e., we require that
$\lambda>0.$ This implies that $Q$ is H\"{o}lder continuous at the origin; see \eqref{candec} below.

Let us fix $\point\in\Int S$, say $\point=0$. Also fix a number $c>-1$,
a positive integer $n$, and a suitable, smooth, real-valued function $\pert$.
Given this, we define an $n$-dependent potential
\begin{equation}\label{pot1}V_{n}(\zeta) = Q(\zeta)+ \tfrac {2c}n\log \tfrac 1{|\zeta|}-\tfrac 1n\pert(\zeta).\end{equation}

We refer to $Q$ as the \textit{underlying potential} - this is what determines the global properties of the droplet.
The logarithmic term in \eqref{pot1} has an effect on the microscopic distribution of particles
near $0$ and near the boundary of the droplet. More precisely, it corresponds to a point charge of strength $c$ at the origin, and at the same time, a harmonic measure on the boundary of the droplet with opposite sign. This effect can be understood as a balayage operation, sweeping of a measure on a domain (here, a point mass in the droplet) to the boundary.
The extra freedom afforded by the term $\pert/n$
is used to
accommodate different kinds of microscopic behaviour near the boundary.

\smallskip

We now define ensembles.
Given $V_n$ of the form \eqref{pot1}, we consider a random system $\{\zeta_j\}_1^n\subset\C$ of $n$ identical point charges
under the influence of the external field $nV_n$.
The system is picked at random with respect to the Boltzmann-Gibbs probability law
\begin{equation}\label{bogi}d\bfP_{n}=\tfrac{1}{Z_{n}} e^{-\Ham_{n}}d\!A_n,\qquad \Ham_n:=\sum_{j\neq k} \log \tfrac{1}{|\zeta_j-\zeta_k|}+n\sum_{j=1}^{n}V_n(\zeta_j)\end{equation}
where $Z_{n}=\int_{\C^n} e^{-\Ham_{n}}\,dA_n$ is the partition function and $dA_n=dA^{\otimes n}$ is suitably normalized Lebesgue measure on $\C^n$.

(The system $\{\zeta_j\}_1^n$ represents the eigenvalues of $n \times n$ random normal matrices equipped with the probability measure proportional to
$e^{-n\mathbf{tr} \, V_n(M)} dM$.
Here, $\bf{tr}$ is the trace and $dM$ is the measure on the space of normal matrices induced by the Euclidean metric on $\C^{n^2}$. See e.g. \cite{CZ} and \cite{EF} for details.)

Recalling that $0\in\Int S$ we shall study the distribution of rescaled system
\begin{equation*}z_j=r_n^{-1}\cdot \zeta_j,\qquad r_n:=n^{-1/2\lambda},\end{equation*}
where $\lambda>0$ is given by \eqref{lam_lim}.
The number $r_n=n^{-1/2\lambda}$ might be called
a \textit{microscopic scale}. (This scale is chosen so that $\sigma(D(0,r_n))=Cn^{-1}+\cdots$ for some positive constant $C$.)

To analyze the large $n$ behaviour of $\{z_j\}_1^n$, we shall
use the \textit{canonical decomposition} of $Q$ about the origin, by which we mean
\begin{equation}\label{candec}Q(\zeta)=Q_0(\zeta)+h(\zeta)+O(|\zeta|^{2\lambda+\epsilon}),\qquad (\zeta\to 0),\end{equation}
where $Q_0$ is a non-negative, homogeneous function of degree $2\lambda$, and where $h$ is a harmonic polynomial of degree at most $2\lambda$ and $\epsilon$ is a positive number. In this article, we assume that the potential has the above canonical decomposition and that the $O$-term remains small after taking Laplacians:
$$\Lap Q=\Lap Q_0+O(|\zeta|^{2\lambda-2+\epsilon}),\qquad (\zeta\to 0).$$
If $Q$ is real analytic at the origin, then we surely obtain the canonical decomposition by Taylor series expansion.
We also assume that $Q_0$ and $\Lap Q_0$ are positive definite, i.e., $Q_0>0$ and $\Lap Q_0>0$ when $z\ne 0$.

We refer to the homogeneous function $Q_0$ as the \textit{dominant part} of $Q$. Typically, we shall find that the rescaled ensemble depends in a ``universal'' way on $Q_0$ and $c$.

\begin{rem} Subtracting an $n$-dependent constant from $Q$ does not change the problem, so we can assume $Q(0)=h(0)=0$. We can likewise assume that $\pert(0)=0$.
In the following, except when otherwise is explicitly stated, we assume that these
normalizations are made.
\end{rem}

\begin{defn}
We say that the point $p=0$ is a \textit{regular} point if $c=0$ and $\lambda=1$; otherwise it is \textit{singular} of type $(\lambda,c)$. If $c\ne 0$ and $\lambda=1$
we speak of a \textit{conical singularity}. We have a \textit{bulk singularity} if $c=0$ and $\lambda\ne 1$,
and a \textit{combined singularity} if $c\ne 0$ and $\lambda\ne 1$.
\end{defn}

We now recall some terminology with respect to the system $\{\zeta_j\}_1^n$ and its rescaled counterpart $\{z_j\}_1^n$.
As a general rule, we designate non-rescaled objects by boldface symbols $\bfR,\bfK,\bfL$, etc., while rescaled objects are written in italics, $R,K,L$, etc.

\smallskip

For a subset $D\subset\C$ we consider the random variable $\calN_D=\#\{j:\zeta_j\in D\}$, and we define,
for distinct $\eta_1,\ldots,\eta_p\in\C$, the \textit{intensity $p$-point function}
$$\bfR_{n,p}(\eta_1,\cdots,\eta_p)=
\lim_{\eps\to 0}\tfrac 1
{\eps^{2p}}\cdot \bfE_n(\prod_{j=1}^p \calN_{\D(\eta_j;\eps)}).$$
The most basic intensity function is the $1$-point function
$\bfR_n(\zeta):=\bfR_{n,1}(\zeta).$

As is well-known (see e.g. the computation in \cite[Section IV.7.2]{ST}), the function $\bfR_{n,p}$ can be expressed as a $p\times p$ determinant
$\bfR_{n,p}(\zeta_1,\cdots,\zeta_p) = \det\left(\bfK_n(\zeta_i,\zeta_j)\right)_{i,j=1}^{p},$
where $\bfK_n$ is a Hermitian function called a correlation kernel of the process. Indeed, we can take $\bfK_n$ as
$$\bfK_n(\zeta,\eta) = \sum_{j=0}^{n-1}p_{n,j}(\zeta)\bar{p}_{n,j}(\eta) e^{-nV_n(\zeta)/2-nV_n(\eta)/2},$$
where $p_{n,j}$ is the orthonormal polynomial of degree $j$ with respect to the measure $e^{-nV_n} dA$.
 We shall always use the symbol $\bfK_n$ to denote this \textit{canonical} correlation kernel.

The rescaled system $\{z_j\}_1^n$ has $p$-point function
$R_{n,p}(z_1,\ldots,z_p)=\det(K_n(z_i,z_j))$
where the kernel $K_n$ is given by
\begin{equation}\label{scala}K_n(z,w)=r_n^2\,\bfK_n(\zeta,\eta),\qquad (z=r_n^{-1}\zeta,\, w=r_n^{-1}\eta).\end{equation}

\smallskip

We now define function spaces. We will employ a function $V_0$, which we call the \textit{microscopic potential} at $0$, whose
definition is
\begin{equation}\label{V0}V_0(z)=Q_0(z)-2c\log|z|.\end{equation}
Here, $Q_0$ is the dominant part of $Q$ in \eqref{candec}.
We also define a corresponding measure $\mu_0$ by
$$d\mu_0=e^{-V_0}\, dA,$$
 and write $L^2_a(\mu_0)$ for the Bergman space (``generalized Fock-Sobolev space'') of all entire functions $f$ of finite norm,
$\|f\|_{L^2(\mu_0)}=(\int_\C |f|^2\, d\mu_0)^{1/2}.$
The Bergman kernel of this space will be denoted by
$L_0(z,w).$

\begin{ex} (``Model Mittag-Leffler ensemble''). Fix three real parameters $\tau_0,\lambda,c$ with $\tau_0,\lambda>0$ and $c>-1$
and define
\begin{equation*}Q_0(\zeta)=\tau_0\cdot \babs{\zeta}^{2\lambda}\quad \text{and}\quad V_{n}(\zeta)=Q_0(\zeta)-\tfrac{2c}{n}\log\babs{\zeta}.\end{equation*}

Rescaling via $\zeta= n^{-1/2\lambda}z$, we obtain \begin{align*}
R_n(z) &= n^{-1/{\lambda}}\, \bfR_n(\zeta) = n^{-1/{\lambda}} \sum_{j=0}^{n-1} |p_{n,j}(\zeta)|^2 e^{-nV_n(\zeta)}.
\end{align*}
A straightforward computation of the
orthonormal polynomials $p_{n,j}$ shows that
$$ R_n(z) = \lambda\tau_0^{(1+c)/\lambda} \sum_{j=0}^{n-1} \ \frac{(\tau_0^{1/\lambda}\babs{z}^{2})^j}{\Gamma(\frac{j+1+c}{\lambda})} e^{-\tau_0\babs{z}^{2\lambda}+2c\log\babs{z}},$$
whence $R_n(z)$ converges to the limit
\begin{equation}\label{mlob}R(z)= \lambda\tau_0^{(1+c)/\lambda}\cdot E_{1/\lambda,(1+c)/\lambda}(\tau_0^{1/\lambda}|z|^2)\cdot
e^{-\tau_0\babs{z}^{2\lambda}+2c\log\babs{z}},\end{equation}
where $E_{a,b}$ is the \textit{Mittag-Leffler function} (cf. \cite{GKMR})
\begin{equation}\label{ml}E_{a,b}(z)=\sum_{j=0}^\infty\frac {z^j}{\Gamma(aj+b)}.\end{equation}
For example, when $\tau_0=\lambda=c=1$ we obtain the well-known one-point function $R(z)=1-e^{-|z|^2}$, cf. \cite[Section 7.6]{AHM2}.

In the present case $V_0$, the microscopic potential, is just $V_0(z)=\tau_0|z|^{2\lambda}-2c\log|z|$, and $\mu_0=e^{-V_0}\, dA$.
It is easy to compute the Bergman kernel $L_0$ of the space $L^2_a(\mu_0)$.
Indeed, using that $L_0(z,w)=\sum_0^\infty e_j(z)\bar{e}_j(w)$ where
$e_j$ are the orthonormal polynomials with respect to $\mu_0$, one finds
\begin{equation}\label{lnoll}L_0(z,w)=\lambda\cdot \tau_0^{(1+c)/\lambda}\cdot E_{1/\lambda,(1+c)/\lambda}(\tau_0^{1/\lambda}z\bar{w}).\end{equation}

Comparing \eqref{mlob} and \eqref{lnoll} we find the basic relation $R(z)=L_0(z,z)e^{-V_0(z)}$ between the microscopic density and the Bergman-kernel
in Fock-Sobolev space. We shall see that this relation remains true in a much more general situation.
\end{ex}

\subsection{Main results} We now state our two principal results, which concern singular points with a certain ``rotational symmetry'', which holds up to a harmonic polynomial.

More precisely, we shall suppose that $Q$ and $V_n$ are of the form
\begin{equation}\label{sqv}Q(\zeta)=Q_r(\zeta)+P(\zeta)\qquad \text{and}\quad V_n(\zeta)=Q(\zeta)+\tfrac {2c}n\log\tfrac 1 {|\zeta|}\end{equation}
where $Q_r$ is radially symmetric and $P$ is a harmonic polynomial. Let $\lambda=1+\lim_{\zeta\to 0}\tfrac {\Delta\log Q_r(\zeta)}{2\log|\zeta|}$ as in \eqref{lam_lim}.

The positive homogeneous part $Q_0$ of degree $2\lambda$
in the canonical decomposition $Q=Q_0+h+\ldots$ is then radially symmetric, and is therefore given by
 $Q_0(z)=\tau_0|z|^{2\lambda}$ for some constant $\tau_0>0$.

Our first main result asserts that the kernel \eqref{lnoll} appears universally at this kind of singular point.

\begin{mth}\label{mt22} Suppose that $Q$ and $V_n$ are of the form \eqref{sqv} and that $Q_0$ is given by $Q_0(z)=\tau_0|z|^{2\lambda}$.

If $S$ is connected and if the outer boundary $\d\Pc S$ is everywhere regular, then the point-processes $\{z_j\}_1^n$ converge as $n\to\infty$ to a unique determinantal point field
with correlation kernel $K(z,w)=L_0(z,w)e^{-V_0(z)/2-V_0(w)/2}$ where $L_0$ is explicitly given as the Mittag-Leffler kernel \eqref{lnoll}.
 \end{mth}

\begin{rem}
The topological assumptions on $S$ in Theorem \ref{mt22} are made for convenience, in order to apply the main result on orthogonal polynomials from \cite{HW}.
\end{rem}

Our next result is a central limit theorem for the normalized Coulomb potential generated by the eigenvalues of a random normal matrix ensemble with respect to
an \textit{algebraic} underlying potential, namely a potential of the special form (near the droplet)
\begin{equation}\label{algo}Q(\zeta)=|\zeta|^{2\lambda}+P(\zeta)\end{equation}
where $P$ is a harmonic polynomial. Correspondingly we define the insertion potential
\begin{equation}\label{hspot}V_n(\zeta)=Q(\zeta)+\tfrac {2c} n\log\tfrac 1 {|\zeta|}.
\end{equation}

This definition requires some extra care, since it may well
happen that the growth condition \eqref{gro} fails (if the degree of $P$ is large). The standard way to circumvent this problem is by considering \textit{local droplets} as in e.g. the papers \cite{EF,LM}, i.e., we redefine $Q$ to be $+\infty$ outside some large compact set.

To be precise, given any (non-polar) compact set
$\Sigma\subset\C$, we may redefine $Q$ to be $+\infty$ outside of $\Sigma$.
The redefined potential has a well-defined droplet $S=S_{Q,\Sigma}$. In the following we shall fix a compact set $\Sigma$ large enough that $S\subset \Int \Sigma$, and then
redefine the potential $Q$ in \eqref{algo} as $+\infty$ outside of $\Sigma$. (The particular choice of such a set $\Sigma$ will be immaterial.)
We shall always adopt this convention about algebraic potentials in the sequel.

 Denote by $\{\zeta_j\}_1^n$ a random sample picked with respect to \eqref{bogi} with $Q$ in \eqref{algo} and write  $\ell(\zeta)=\log|\zeta|$.
It is natural to define a random variable $\tr_n \ell$ by
$$\tr_n\ell=\sum_{j=1}^n\ell(\zeta_j).$$
We shall study the fluctuation about the mean,
$\fl_n \ell:=\tr_n \ell-\bfE_n \tr_n \ell.$ When $0$ is a regular
point, it is expected that $\fl_n\ell$
be in some sense a good approximation to the Gaussian free field evaluated at $0$
provided
that $n$ is large, see e.g. the concluding remarks in \cite{AHM2} or
\cite[Appendix 6]{KM0}.

A problem with making this rigorous is that the variance of the fluctuations grows logarithmically in $n$. There are various ways one could try to circumvent this difficulty; our approach here is to consider the ``normalized'' fluctuations
$$X_n=\tfrac 2 {\sqrt{\log n}}(\tr_n\ell-\bfE_n\tr_n\ell).$$

\begin{mth} \label{mt7} Suppose that $Q$ is algebraic, of the form \eqref{algo}. Let $\{\zeta_j\}$ be picked randomly with respect to $V_n$ in \eqref{hspot}.
Suppose also that the conditions of Theorem \ref{mt22} are satisfied.
Then $X_n$ converges in distribution to the normal distribution with mean $0$, variance $1/\lambda$.
\end{mth}

 The special case of Theorem \ref{mt7} when
 $V_n(\zeta)=|\zeta|^2+2\re t\zeta-2(c/n)\log|\zeta|$, $(|t|<1)$ follows from the work of Webb and Wong in \cite[Corollary 1.2]{WW}.

 We also note that the distribution of random variables $$\tilde{X}_n(f):=\frac {\tr_n f-\bfE_n\tr_n f}{\sqrt{\Var_n(\tr_n f)}},
 \qquad \tr_n f:=\sum_1^n f(\zeta_i),$$ has been well-studied when $f$ is, for example, a characteristic function. In fact,
 Soshnikov in \cite[Theorem 1]{S2} has shown that $\tilde{X}_n(f)$ converges in distribution to the standard normal for a large class of \textit{bounded} test functions $f$ such that the variance $\Var_n(\tr_n f)\to\infty$ as $n\to \infty$. Theorem \ref{mt7} gives a similar result for the random variables $\tilde{X}_n(f)$ in the case where $f=\ell$ is unbounded. Indeed, our result immediately extends to a central limit theorem for the random fields $Y_n(p)=\trace_n\ell_p$, $\ell_p(\zeta)=\log|\zeta-p|$ with $p$ in the bulk. This is consistent with the hypothesis that $Y_n$ comes close to a $\log$-correlated Gaussian field inside the droplet.
We refer to \cite{La20} for a more thorough discussion of the relationship between $Y_n$, the Gaussian free field, and multiplicative chaos theory in the special case of the Ginibre ensemble. Recently, the fluctuations when $f$ is a characteristic function have been studied in \cite{FL} for the Ginibre ensemble.
Other relevant references in the context of unitary invariant random matrix ensembles are \cite{CFLW,FHK,HKO01,LP19}.

We remark also that cases of algebraic
potentials
have appeared in several works, e.g.~ \cite{BEG,BS,KT,LM},
while \cite{LY,WW} investigate the insertion potential
$V_n=|\zeta|^2+2\re (t\zeta)+\tfrac {2c}n\log\tfrac 1 {|\zeta|}$, ($|t|<1$, $c>-1$).

\subsection{Outline of the strategy; further results}

In the following, we assume that the general conditions on the potentials $Q$ and $V_n$ in Subsection \ref{bset}
are satisfied.

Recall that a Hermitian function $K$ is called a \textit{positive matrix} if the quadratic form $\sum_{j,k=1}^{N}\alpha_j \bar{\alpha}_k K(z_j,z_k)$ is non-negative for all scalars $\alpha_i \in \C$ and points $z_i \in \C$.

Our starting point is the following structure theorem for limiting kernels.

\begin{mth} \label{mt1} Let $K_n$ be the rescaled canonical correlation kernel. There exists a sequence of cocycles $c_n$ such that
\begin{equation}\label{mcon}c_n(z,w)K_n(z,w)=L_n(z,w)e^{-V_0(z)/2-V_0(w)/2}(1+o(1))\end{equation}
where $L_n(z,w)$ is Hermitian-entire, $o(1)\to 0$ locally uniformly
on $\C^2$ as $n\to\infty$.

The kernels $\{L_n\}$ have the compactness property that each subsequence
has a further subsequence converging
locally uniformly on $\C^2$ to a Hermitian-entire function $L$ which satisfies
the mass-one inequality:
\begin{equation}\label{mass}\int|L(z,w)|^2\, d\mu_0(w)\le L(z,z).\end{equation}
Moreover, $L$ is a positive matrix and the inequality
$L\le L_0$
holds in the sense of positive matrices, where $L_0$ is the Bergman kernel of $L^2_a(\mu_0)$.
\end{mth}

The kernel $L$ in the theorem is called a \textit{limiting holomorphic kernel}. We define the
corresponding \textit{limiting correlation kernel} $K$ by
$$K(z,w)=L(z,w)e^{-V_0(z)/2-V_0(w)/2}.$$
We also write
$R(z)=K(z,z)$
and speak of a \textit{limiting 1-point function} or \textit{microscopic density}.

We have the following ``tightness theorem''.

\begin{mth} \label{m1.2} Each limiting kernel is
nontrivial, in the sense that $R$ does not vanish identically.
More precisely,
there exists a constant $\alpha>0$ such that
\begin{align}\label{abc}R(z)&=\Lap Q_0(z)\cdot (1+O(e^{-\alpha|z|^{2\lambda}})), &(z\to\infty)\\
\label{def} R(z)&=O(|z|^{2c}), &(z\to 0).
\end{align}
The kernel $K$ is the correlation kernel of a unique point field which is
determined by the collection of $p$-point functions
$R_p(z_1,\ldots,z_p)=\det (K(z_i,z_j))_{i,j=1}^p$. The point field is the limit of the ensemble $\{z_j\}_1^n$
in the sense of point processes as $n\to\infty$. \end{mth}

Here ``convergence in the sense of point processes'' means convergence of all $k$-point functions, cf. \cite[Sect. 1.5]{AKMW} and the references therein.

Note the meaning of the asymptotic in \eqref{abc}: the first intensity  quickly approaches the classical equilibrium density $\Lap Q$ as
one moves away from the singular point. See Figure 1.

\begin{ex} Consider the model Mittag-Leffler ensemble, $V_n(\zeta)=|\zeta|^{2\lambda}+\tfrac {2c} n\log\tfrac 1 {|\zeta|}$.

It is interesting to compare the asymptotic in \eqref{abc} with expansion
formulas for Mittag-Leffler functions in \cite{GKMR}.
To this end, we note that if $\lambda>1/2$ and if $|z|^2>\eps>0$ we have by \cite[eq. (4.7.2)]{GKMR}
\begin{equation}\label{mop}R(z)-\Lap Q_0(z)=\lambda^2|z|^{2c}e^{-|z|^{2\lambda}}\cdot\frac {1}{2\pi i}\int_{\gamma(\eps,\delta)}
\frac {e^{\zeta^\lambda}\zeta^{\lambda-1-c}}{\zeta-|z|^2}\, d\zeta,\end{equation}
where $\delta$ is some number in the interval $\pi/2\lambda<\delta<\pi/\lambda$ and $\gamma(\eps,\delta)$ is the contour
consisting of the two rays $S_{\pm\delta,\eps}=\{\arg\zeta=\pm\delta,\,|\zeta|\ge\eps\}$ and the circular arc
$C_{\delta,\eps}=\{|\zeta|=\eps,\, |\arg\zeta|\le\delta\}$.

The right hand side in \eqref{mop} has the asymptotic expansion, as $z\to\infty$:
$$-\lambda|z|^{2c}e^{-|z|^{2\lambda}}(\tfrac 1 {\Gamma(c/\lambda)}|z|^{-2}
+\tfrac 1 {\Gamma((c-1)/\lambda)}|z|^{-4}+\cdots).$$ (Cf. \cite[eq. (4.7.4)]{GKMR}.)
By contrast, the asymptotic in \cite[Theorem 4]{AS2} holds for all $\lambda>0$ (and for more general potentials) but gives less precise
information when $\lambda>1/2$.

\begin{figure}[ht]\label{m_graph}
\begin{center}

\includegraphics[width=.3\textwidth]{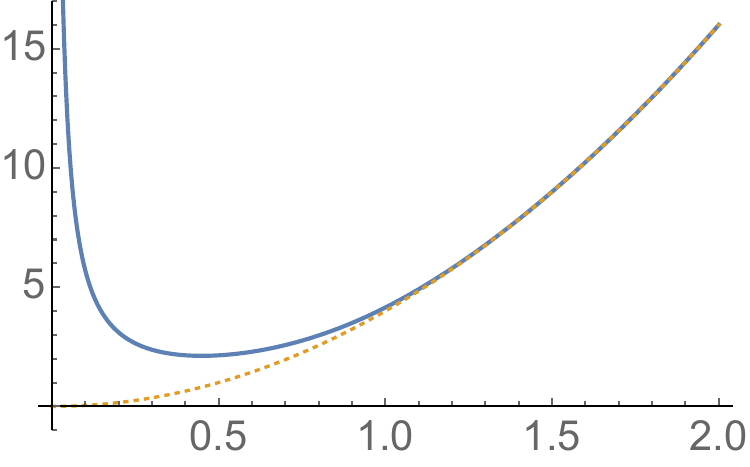}
\hspace{.03\textwidth}
\includegraphics[width=.3\textwidth]{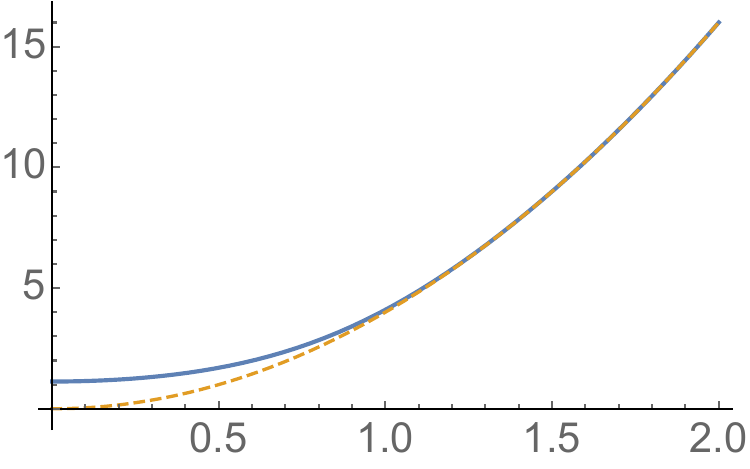}
\hspace{.03\textwidth}
\includegraphics[width=.3\textwidth]{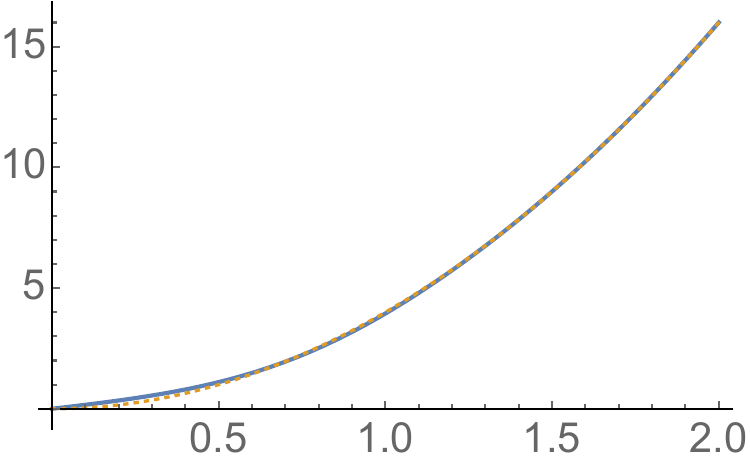}
\end{center}
\caption{The graph of $R$ restricted to the positive real axis for combined singularities of type $(c,2)$, with $c=-0.5$, $c=0$, and $c=0.5$. For comparison, the graph of $\Lap Q_0$ is drawn with an orange dashed line.}
\end{figure}

If $c$ is negative, the inserted charge is attractive and
$R=+\infty$ at the location of the charge, while if $c$ is positive, the inserted charge is repulsive and $R=0$ there.  (See Figure 1 as well as Section \ref{sec5}).

We note that the special case of Mittag-Leffler ensembles where $\lambda=k$ is an integer and $c=0$ was (except for the name) considered by Chau and Zaboronsky in \cite[Section 3]{CZ}.
The limit kernel is there expressed in terms of a Kummer function $_1F_1$ rather than the Mittag-Leffler function $E_{1/k,1/k}$. The potential $Q=|\zeta|^{2\lambda}$ also appears in \cite{AB}
under the name ``Freud potential''.
\end{ex}

\smallskip

To describe a limiting point field, we must determine the corresponding limiting kernel $L$ or, what amounts to the same thing,
the correlation kernel $K$. Our next result, \textit{Ward's equation}, gives some general information about these kernels.
To formulate this result we define (at points where $R(z)\ne 0$) the \textit{Berezin kernel}
$$B(z,w)=\frac {|K(z,w)|^2}{R(z)}$$
as well as the \textit{Cauchy transform}
\begin{equation*}C(z)=\int_\C \frac {B(z,w)}{z-w}\, dA(w).\end{equation*}
The Berezin kernel $B(z,w)$ expresses the intensity of repulsion from a particle located at $z$. More precisely, if $K$ is a correlation kernel of a point process $\Phi$ and $R$ is its $1$-point function, then $B(z,w) = R(w) - \tilde{R}_{z}(w)$, where $\tilde{R}_{z}$ is the $1$-point function for the conditional point process $\Phi_{z}$ given that a point of $\Phi$ is located at $z$. See \cite[Section 7.6]{AHM2} for more  detailed discussion on the Berezin kernel.

We now have the ingredients to formulate Ward's equation.

\begin{mth} \label{m1.5} Keep the assumptions from Subsection \ref{bset} except that the origin is not assumed to necessarily be in the bulk of $S$.
Then each limiting $1$-point function $R$ is either trivial ($R=0$ identically) or is
strictly positive on $\C^*$ and
$$\dbar C=R-\Lap V_0-\Lap\log R$$
pointwise on $\C^*$ and in the sense of distributions on $\C$.
\end{mth}

Combining with Theorem \ref{m1.2} we see that if $0$ is assumed to be in the bulk, then $R$ in Theorem \ref{m1.5} is non-trivial, that is, $R>0$ at each point of $\C^*$ and Ward's equation
is satisfied.

We stress that $C$ is uniquely determined by $R$, so Ward's equation gives a feedback relation for the sole unknown function $R$. Also it should be noted that Ward's equation itself typically has many solutions. To guarantee uniqueness we need to supply suitable apriori conditions on the density $R$, which depend on the nature of the zooming-point. For example, the same Ward equation holds whether we zoom on a bulk point or a boundary point, and yet the densities look very different.

Our next result concerns the case where
the dominant part $Q_0$ in the decomposition \eqref{candec}
satisfies $Q_0(\zeta)=Q_0(|\zeta|)$. We can then find $\tau_0$ and $\lambda$ such that $Q_0(z)=\tau_0|z|^{2\lambda}$.

In this case it is natural
 to expect that a limiting kernel $L$ be \textit{symmetric} in the sense that
\begin{equation*}L(z,w)=E(z\bar{w})\end{equation*}
for some entire function $E$.

\begin{mth} \label{mt2}
If $Q_0(z)=\tau_0|z|^{2\lambda}$ then each
symmetric limiting holomorphic kernel $L$ equals to $L_0$, where $L_0$ is given by \eqref{lnoll}.
\end{mth}

From a naive point of view it might seem obvious that each limiting kernel in Theorem \ref{mt2} should be symmetric.
The question is however quite subtle; it is connected to the vanishing of a certain entire function $G(z)$, as explained
in Section \ref{ud}. However, if the potential is algebraic (and the droplet is connected), then we are in fact able to show that each limiting kernel is symmetric. Hence Theorem \ref{mt22} will follow
as a consequence of Theorem \ref{mt2}.

\subsection{Yet further results and plan of the paper} In Section \ref{gener}, we use estimates
from the companion paper \cite{AS2} to prove
Theorem \ref{mt1}.

In Section \ref{zw} we prove Theorem \ref{m1.2}. We then introduce the distributional Ward equation and prove
Theorem \ref{m1.5}.

In Section \ref{ud} we prove Theorem \ref{mt2} and analyze Ward's equation in
the radially symmetric case, thus extending the analysis from the papers \cite{AK,AS}.

In Section \ref{ginen} we complete the analysis of the case when the dominant part $Q_0$ is radially symmetric by proving
apriori symmetry under the hypotheses in Theorem \ref{mt22}.

In Section \ref{sec5} we study the effect of inserting a point charge, by comparing the $1$-point functions with and without insertion. It turns out that the difference
gives rise to a balayage operation, taking mass from the insertion and distributing it near the boundary, according to a harmonic measure.

In Section \ref{CLT} we prove Theorem \ref{mt7}, modulo an estimate of the 1-point function which is postponed to Section \ref{asyl}.

In Section \ref{asyl} we prove that, for algebraic underlying potentials, the asymptotics $\bfR_n(\zeta)\sim
n\lambda^2|\zeta|^{2\lambda-2}$ as $n\to\infty$ holds to within a very small error,
when $\zeta\in\Int S$ is far enough away from the singular point as well as from the boundary. This result can be regarded as a Tian-Catlin-Zelditch type expansion for algebraic insertion potentials, which we were unable to find in the existing literature. We will
use the result to complete our proof of Theorem \ref{mt7}.

\section{Limiting kernels and their basic properties} \label{gener}
In this section we prove Theorem \ref{mt1} on the structure of limiting kernels.

\smallskip

We start with Theorem \ref{mt1}. The proof will follow easily from the estimates in \cite{AS2}, once
the proper notation has been introduced.

\smallskip

Write $d\mu_n=e^{-nV_n}\, dA$, (see \eqref{pot1}), and let $\calP_n=\Pol(n)$
equipped with the norm of $L^2(\mu_n)$. Writing $\bfk_n$ for the reproducing kernel
of $\calP_n$, we have
$$\bfR_n(\zeta)=\bfk_n(\zeta,\zeta)e^{-nV_n(\zeta)}.$$
(This follows by the classical Dyson determinant formula, which goes back to the early days of random matrix theory, cf. \cite{M}.)

Now rescale: let $r_n=n^{-1/2\lambda}$ and put $z=r_n^{-1}\zeta$, $w=r_n^{-1}\eta$. We define
$$k_n(z,w)=r_n^{2+2c}\bfk_n(\zeta,\eta).$$

It is convenient to introduce a ``rescaled potential'' by
$$\tilde{V}_n(z)=nQ(r_nz)-2c\log|z| -\pert(r_nz).$$
Using this notation, the kernel $K_n$ in \eqref{scala} can be expressed as
$$K_n(z,w)=k_n(z,w)e^{-\tilde{V}_n(z)/2-\tilde{V}_n(w)/2}.$$

Next recall the canonical decomposition \eqref{candec}. We  recognize that $h$ is the real part of the holomorphic polynomial $H$ whose degree is a positive integer $2d \leq 2\lambda$.
(Recall that we have assumed that $Q(0)=H(0)=\pert(0)=0$.)

Now consider the factorization
$$K_n(z,w)=L_n(z,w)\cdot E_n(z,w)\cdot |zw|^c$$ where $E_n$ and $L_n$ are defined by
\begin{equation}\label{Ln}
\begin{cases}E_n(z,w) &= e^{n(H(r_n z)+\bar{H}(r_n w)-Q(r_nz)-Q(r_nw))/2+(\pert(r_nz)+\pert(r_nw))/2},\cr
& \cr
L_n(z,w) &=k_n(z,w)\,e^{-nH(r_n z)/2-n \bar{H}(r_n w)/2}.\cr
\end{cases}
\end{equation}

Note that $L_n$ is Hermitian-entire while, by Taylor's formula,
\begin{equation*}R_n(z)=L_n(z,z)E_n(z,z)|z|^{2c}=L_n(z,z)e^{-V_0(z)(1+o(1))}
\end{equation*}
where $V_0$ is the microscopic potential in \eqref{V0} and
$o(1)\to 0$ as $n\to\infty$, uniformly on compact subsets of $\C$.

\begin{lem} \label{hlem} Let $f_n(z,w)$ be a sequence of Hermitian-entire functions such that $f_n\to 0$ locally
uniformly on $\C^{*2}$. Then $f_n\to 0$ locally uniformly on $\C^2$.
\end{lem}

\begin{proof}
Use Cauchy's formula
$f_n(z,w)=\frac 1 {(2\pi i)^2}\iint_{|u|=|v|=R}\frac {f_n(u,v)}{(u-z)(\bar{v}-\bar{w})}\, dudv$ for large enough $R$.
\end{proof}

By \cite[Corollary 4.7]{AS2}, each subsequence of $\{L_n\}$ has a further
 subsequence (renamed as $L_n$) such that $L_n\to L$ locally uniformly on $\C^{*2}$ where
$L$ is Hermitian-analytic and locally bounded. It is easy to see that $L$ extends to a Hermitian-entire function. Indeed, for fixed $w\ne 0$ the function $z\mapsto L(z,w)$ has a removable singularity
at $0$, by Riemann's theorem on removable singularities. Likewise for the functions $w\mapsto L(z,\bar{w})$ with $z\ne 0$. Thus $L$ extends
to $\C^2\setminus \{0\}$, and hence to $\C^2$ by the Hartogs' theorem (see \cite[Theorem 2.10.1]{K}). In this way,
we always regard $L$ as a Hermitian-entire function in the sequel.

\begin{lem} We have that $L_n\to L$ locally uniformly on $\C^2$.
\end{lem}

\begin{proof} Apply Lemma \ref{hlem} to the difference $f_n=L_n-L$.
\end{proof}

The mass-one inequality \eqref{mass} now follows from the following argument. Let us write the canonical decomposition \eqref{candec} as
$$Q=Q_0+\re H+Q_1,$$ where $Q_1=O(|\zeta|^{2\lambda+\epsilon})$ as $\zeta\to 0$ for some $\epsilon>0$.

Define a measure $\mu_{0,n}$ by
$$d\mu_{0,n}(z)=e^{-V_0(z)-nQ_1(r_nz)+\pert(r_nz)}\, dA(z),$$
and note that the kernel $L_n$ has the following reproducing property,
$$\int_\C |L_n(z,w)|^2\, d\mu_{0,n}(w)=L_n(z,z).$$
The mass one inequality \eqref{mass} follows from this and Fatou's lemma, since $\mu_{0,n}\to\mu_0$ in the vague sense of measures where $d\mu_0=e^{-V_0}\, dA$.
(By ``vague convergence'', we mean that $\mu_{0,n}(f)\to \mu_0(f)$ as $n\to\infty$ for each continuous, compactly supported function $f$.)

\smallskip

We now recognize $L_n(z,w)$ as the reproducing kernel for the Hilbert space $\calH_n$ of entire functions defined by
\begin{equation*}
\calH_n = \{ f(z)=p(z)\cdot e^{-nH(r_nz)/2}\ ;\ p\in \mathrm{Pol}(n) \}\end{equation*}
with the norm of $L^2(\mu_{0,n}).$

Since $L_n\to L$ and since $\mu_{0,n}\to\mu_0$ vaguely, it follows by the arguments in \cite[Section 2.3]{AS} (or see \cite[Section 3.7]{AKM}, \cite{AS2}) that
$L$ is the Bergman kernel of some semi-normed Hilbert space $\calH_*$ of entire functions, which sits
contractively inside $L^2_a(\mu_0)$ (i.e., the natural inclusion is a contraction). Hence $L$ is a positive matrix. Moreover, by Aronszajn's theorem on differences of reproducing kernels in \cite[Part I.7]{Ar},
the difference $L_0-L$ is a positive matrix, i.e., we have $0\le L\le L_0$.

\smallskip

We now turn to the kernels $E_n$ from
\eqref{Ln}. By Taylor's formula (since we have assumed $\pert(0)=0$)
$$E_n(z,w)=e^{-Q_0(z)/2-Q_0(w)/2}\cdot e^{i\im (nH(r_nz)-nH(r_nw))/2}\cdot(1+o(1)),$$
where the second factor in the right hand side is a cocycle. Letting $c_n(z,w)$ be the reciprocal cocycle,
we thus have the convergence
$$c_n(z,w)K_n(z,w)=L_n(z,w)e^{-V_0(z)/2-V_0(w)/2}(1+o(1)),$$
where $o(1)\to 0$ locally uniformly as $n\to\infty$.
By this, Theorem \ref{mt1} is proved. q.e.d.

\section{Zero-one law and Ward's equation} \label{zw}
Our main goal with this section is to verify Theorem \ref{m1.5}.
We start however with Theorem \ref{m1.2}.

\begin{proof}[Proof of Theorem \ref{m1.2}] The estimate for $R(z)$ as $z\to\infty$ in \eqref{abc} is proved in \cite[Theorem 4]{AS2}.
To prove \eqref{def},
we start by noting that the $1$-point functions
$$R_n(z)=K_n(z,z)=L_n(z,z)e^{-V_0(z)}(1+o(1))$$
obey a uniform bound of the form $R_n(z)\le M|z|^{2c}$ when $|z|\le 1$. Indeed, this follows from Theorem
\ref{mt1}, via the relation \eqref{mcon}, on noting that the functions $L_n(z,z)$ which converge uniformly
to $L(z,z)$ on the unit disc, remain uniformly bounded there.

The statement about convergence of point fields follows from the upper bound $R(z)\le L_0(z,z)e^{-V_0(z)}$, since this guarantees that $R_n\to R$ in
$L^1_{\text{loc}}$
(cf. \cite[Lemma 1]{AKMW}).
\end{proof}

\subsection{A distributional Ward identity} \label{dward}
We are first going to verify Ward's identity (or loop equation) in the sense of distributions. To this end,
it is convenient to use the integration by parts approach from \cite[Section 4.2]{BBNY2}.

As usual, we write
$$\Ham_n=\sum_{j\ne k}\log\tfrac 1 {\babs{\,\zeta_j-\zeta_k\,}}+n\sum_{j=1}^nV_n(\zeta_j).$$
Let $\psi$ be a test-function. Interpreting the $\d$-derivative in the sense of distributions, we have (for all $j$)
\begin{equation}\label{enn}\bfE_n[{\d} \psi(\zeta_j)]= \bfE_n[\d_j\Ham_n(\zeta_1,\ldots,\zeta_n)\cdot \psi(\zeta_j)],\end{equation}
where $\d_j=\d/\d\zeta_j$ and where $\bfE_n$ is expectation with respect to the Boltzmann-Gibbs law in \eqref{bogi}.

Summing over $j$ in \eqref{enn} gives
\begin{align*}\tfrac 1 n\sum_{j=1}^n\bfE_n[{\d}\psi(\zeta_j)]&=
 \bfE_n\sum_{j=1}^n \psi(\zeta_j)({\d}V_n(\zeta_j)-\tfrac{1}{n}\sum_{k\ne j}\tfrac 1 {\zeta_j-\zeta_k})\end{align*}
We have shown that
\begin{equation}\label{ward:id}\bfE_n[W_n^+[\psi]]=0\end{equation}
where
$$W_n^+[\psi]=\sum \d\psi(\zeta_j)-n\sum[\psi\d V_n](\zeta_j)+ \tfrac 1  2 \sum_{j\ne k}\tfrac {\psi(\zeta_j)-\psi(\zeta_k)}
{\zeta_j-\zeta_k}.$$
This is the distributional form of Ward's identity that we need; the point is that $\psi$ is
compactly supported in $\C$, not just in $\C^*$.

\subsection{Rescaling}

(Cf. \cite{AKM,AS}.)
Fix a test function $\psi \in C_0^{\infty}(\C)$ and observe that
$$W_n^{+}[\psi]=\I_n[\psi]-\II_n[\psi]+\III_n[\psi]$$
where
\begin{align*} &\I_n[\psi] = \tfrac{1}{2}\sum_{j\neq k}^{n} \tfrac{\psi(\zeta_j)-\psi(\zeta_k)}{\zeta_j-\zeta_k}, \quad \II_n[\psi]=n\sum_{j=1}^{n}\d V_{n}(\zeta_j) \cdot \psi(\zeta_j),\quad\mathrm{and}\\
&\III_n[\psi]= \sum_{j=1}^{n} \d \psi(\zeta_j).
\end{align*}

Rescaling via $z =r_n^{-1}\zeta$ and $w=r_n^{-1}\eta$, we define the Berezin kernel and its Cauchy transform by
\begin{align*}
B_n(z,w)&=\tfrac {R_n(z)R_n(w)-R_{n,2}(z,w)}
{R_n(z)},\cr
C_n(z)&=\int_\C\tfrac {B_n(z,w)}{z-w}\, dA(w).
\end{align*}

\begin{lem}\label{lem:ward}
We have
\begin{equation*} \bar{\d} C_n(z) = R_n (z) - \Lap V_0(z) -\Lap \log R_n(z) + o(1)
\end{equation*}
where $o(1)\to 0$ in the sense of distributions on $\C$ and uniformly on each compact subset of $\C^*$ as $n\to\infty$.
\end{lem}
\begin{proof}
Let $\psi_n(r_nz)=\psi(z)$.
By \eqref{ward:id} we obtain
$$\bfE_n W_n^{+}[\psi_n]= \bfE_n \I_n[\psi_n]-\bfE_n \II_n[\psi_n]+\bfE_n \III_n[\psi_n]=0.$$
By changing variables, we calculate each expectation as follows:
\begin{align*}
\bfE_n \I_n[\psi_n]&= r_n^{-1} \int_\C \psi(z) dA(z)\int_\C \tfrac{R_{n,2}(z,w)}{z-w} dA(w),\\
\bfE_n \II_n[\psi_n]&= n \int_\C \d V_{n}(r_n z) \psi(z) R_{n}(z) dA(z),\\
\bfE_n \III_n[\psi_n]&= r_n^{-1} \int_\C \d\psi(z) R_{n}(z) dA(z)= -r_n^{-1} \int_\C \psi(z) \d R_{n}(z) dA(z).
\end{align*}

From the above, we obtain that
$$ \int_\C \tfrac{R_{n,2}(z,w)}{z-w} dA(w) = n r_n \d V_{n}(r_n z) R_{n}(z) + \d R_{n}(z)$$
in the sense of distributions on $\C$. Since $$R_{n,2}(z,w)=R_{n}(z) (R_{n}(w)-B_{n}(z,w)),$$ we have
$$ \int_\C \tfrac{B_n(z,w)}{z-w} dA(w) = \int_\C \tfrac{R_{n}(w)}{z-w} dA(w) - n r_n \d V_{n}(r_n z) - \d \log R_{n}(z).$$
This equation holds pointwise on $\C^*$ and in the sense of distributions on $\C$.
Differentiating both sides in the sense of distributions with respect to $\bar{z}$, we have
$$\bar{\d} C_n(z) = R_{n}(z) - nr_n^2 \Lap V_{n}(r_n z) - \Lap \log R_{n}(z).$$
It remains to note that, in the sense of distributions,
$$nr_n^2\Lap V_n(r_nz)=c\delta_0(z)+\Lap Q_0(z)+O(r_n)=\Lap V_0(z)+O(r_n),$$
where the $O$-constant
is also uniform on each compact subset of $\C^*$.
\end{proof}

\subsection{Ward's equation} We now observe
that the Berezin kernel $B_n$ can be written as
\begin{equation}\label{jugo}
B_n(z,w) = \tfrac{\babs{K_n(z,w)}^2}{R_n(z)}=\tfrac{\babs{L_n(z,w)\cdot E_n(z,w)}^2}{L_n(z,z)\cdot E_n(z,z)} |w|^{2c}
\end{equation}
(See \eqref{Ln} for the definitions of the functions $L_n,E_n$.)

In order to prove Ward's equation, it is convenient to use a formulation in terms of
a limiting holomorphic kernel $L$. We remind of the basic relation
$$R(z)=L(z,z)e^{-V_0(z)}.$$

\begin{lem}\label{L:zero}
Let $L$ be a limiting holomorphic kernel and $z_0\in\C$. If $L(z_0, z_0)=0$, then $L(z,z)=\babs{z-z_0}^2 \tilde{L}(z,z)$ for some Hermitian-entire function $\tilde{L}$. Moreover, if $L(z,z)$ is not identically zero, then each zero of $L(z,z)$ is isolated.
\end{lem}

\begin{proof}
If $L(z_0,z_0)=0$, then
\eqref{mass} gives that $\int \babs{L(z_0,w)}^2 d\mu_0(w) = 0$, so $L(z_0,w)=L(w,z_0)=0$, and $L(z,w)=(z-z_0)(\bar{w}-\bar{z}_0)\tilde L(z,w)$ for some Hermitian-entire $\tilde{L}$.

Now assume that $L(z,z)$ is a nontrivial kernel which has a zero $z_0$ which is not isolated, i.e., that there exist distinct zeros $z_j$ such that $z_j \to z_0$. Then $L(z_j,w)=0$ for all
$w$ and all $j$. Noting that $L(z,w)$ is entire in $z$, we obtain a contradiction.
\end{proof}

\begin{lem}\label{logsub}
$z\mapsto L(z,z)$ is logarithmically subharmonic on $\C$. Moreover, if
$R(z_0)=0$ for some $z_0$, and if we put
$$L(z,z)=\babs{z-z_0}^2 g(z),$$
then
$g$ is logarithmically subharmonic some neighbourhood of $z_0$.
\end{lem}

\begin{proof} Recall first that $R$ does not vanish identically by Theorem \ref{m1.2}.

By Theorem \ref{mt1}, we know that $L$ is the Bergman kernel of a contractively embedded, semi-normed
Hilbert space $\calH_*\subset L^2_a(\mu_0)$. (This is because of the inequality $L\le L_0$.)
That $\log L(z,z)$ is subharmonic now follows from general Bergman space theory, as in
\cite[Lemma 3.4]{AS}.

 Next assume that $L(z,z)=\babs{z-z_0}^2 g(z)$ and choose a small neighborhood $D$ of $z_0$ which does not contain any zero of $L(z,z)$ except $z_0$.
Since $\Lap_z\log g(z)= \Lap_z \log L(z,z)-\delta_{z_0}$, $\Lap \log g\geq0$ on $D\setminus\{z_0\}$.

There are now two possibilities:
if $g(z_0) >0$, then we extend $\log g$ analytically to $z_0$ and have $\Lap\log g(z_0)\geq 0$; if $g(z_0)=0$, then $\log g(z_0)=-\infty$ and $\log g(z)$ again satisfies the sub-mean value property in $D$.
\end{proof}

We now set out to find suitable subsequential limits of the Berezin kernels $B_n$, defined in \eqref{jugo}.
For this, we fix
a subsequence $L_{n_\ell}$ which converges locally uniformly to a limiting holomorphic kernel $L$.

The main observation is that if $L(z_0,z_0)>0$, then the convergence
\begin{equation*}B_{n_\ell}(z,w)\cdot |w|^{-2c}\to B(z,w)\cdot |w|^{-2c},\qquad (\ell\to\infty)\end{equation*}
is uniform for all $(z,w)\in D\times K$ where $D$ is some neighbourhood of $z_0$ and $K$ is a given compact subset of $\C$. To see this, it suffices to
note that
\begin{equation*}\begin{cases} B_{n}(z,w)&=\tfrac {|L_{n}(z,w)|^2}{L_{n}(z,z)}|w|^{2c}e^{-Q_0(w)+O(r_n)}\cr
& \cr
B(z,w)&=\tfrac {|L(z,w)|^2}{L(z,z)}|w|^{2c}e^{-Q_0(w)}\cr
\end{cases},
\end{equation*}
and that $L(z,z)=\lim L_{n_\ell}(z,z)\ge \const >0$ in a neighbourhood of $z_0$.

\smallskip

We need to check that the convergence $B_{n_\ell}\to B$ implies a suitable convergence
$C_{n_\ell}\to C$ on the level of Cauchy transforms. For this purpose, we formulate the next lemma.

\begin{lem}\label{C:bound}
Suppose that $R(z)=L(z,z)e^{-V_0(z)}$ does not vanish identically. If $\calZ$ is the set of isolated zeros of $L(z,z)$, then $C_{n_\ell}\to C$ locally uniformly on $\C\setminus (\calZ\cup\{0\})$ as $\ell\to\infty$. Moreover,
the function $z\mapsto zC(z)$ is bounded on $\calV\setminus\calZ$ for each compact subset $\calV\subset \C$.
\end{lem}

\begin{proof} Fix a small number $\epsilon>0$.
We define a compact subset $K_\epsilon$ of $\C^2$ by
$$K_\epsilon=\{(z,w);\, |z|\le 1/\epsilon,\, |w|\le 2/\epsilon,\, \dist(z,\calZ)\ge \epsilon\}.$$

The remarks preceding the lemma show that
we can find $N$ such that if $\ell\geq N$ then
\begin{equation}\label{Bnl}
\babs{B_{n_\ell}(z,w)-B(z,w)}\cdot |w|^{-2c}<\epsilon^{2+2|c|},\quad (z,w)\in K_\epsilon.\end{equation}
Let us recall also that there is a constant $M$ such that, if $z\not\in\calZ$, then
\begin{equation}\label{dih}B_{n}(z,w)\le M|w|^{2c}, \quad B(z,w)\le M|w|^{2c},\quad (|w|\le 1)\end{equation}
by Theorem \ref{m1.2}.
(This is because $B_n(z,w)\le R_n(w)$ and $B(z,w)\le R(w)$.)

For $\ell\geq N$ and $z$ with $\dist(z, \calZ\cup\{0\})\geq\epsilon$ and $\babs{z}\leq 1/\epsilon$, we now obtain
\begin{align*}
&|C_{n_\ell}(z)-C(z)| \leq \int_\C |\tfrac{B_{n_\ell}(z,w)-B(z,w)}{z-w}| dA(w)\\
& =
\int_{\babs{w}<\tfrac{\epsilon}{2},\babs{z-w}<\tfrac{1}{\epsilon}} + \int_{\babs{z-w}<\tfrac{1}{\epsilon}, \babs{w}>\tfrac{\epsilon}{2}} +\int_{\babs{z-w}>\tfrac{1}{\epsilon}}\cdots
\\
&\leq 2 \epsilon^{1+2|c|}\! \int_{\babs{w}<\tfrac{\epsilon}{2}} \babs{w}^{2c}dA(w)+  2^{2|c|}\epsilon^2\int_{\babs{z-w}<\tfrac{1}{\epsilon}}\tfrac{dA(w)}{\babs{z-w}}+2\epsilon\\
&\leq 2^{2|c|-1} (1+c)^{-1} \epsilon^3 + 2^{2|c|+1}\epsilon+2\epsilon.
\end{align*}
Here, we have applied the mass-one inequality \eqref{mass} to estimate the integral over $\{\babs{z-w}>\tfrac{1}{\epsilon} \}$ and the inequality \eqref{Bnl} has been used for the estimate of the integral over $\{|z-w|<1/\epsilon\}$.

To show the local boundedness of the Cauchy transform $C(z)$, we fix a compact subset $\calV$ of $\C$
and a number $\epsilon=\epsilon_{\calV}$ such that $\calV\subset\{|z|<1/\epsilon\}$. Let us write $c'=-\min\{c,0\}$.

{Let $\delta$ be an arbitrary small number.} Using the estimate \eqref{dih}, we see that for $z\in
\calV\setminus(\calZ\cup\{0\})$ with $|z|\ge \delta$,{
\begin{align*}
\babs{C_{n_\ell}(z)}
&\leq ( \int_{\babs{z-w}<\delta/2} + \int_{\babs{z-w}>\delta/2} ) \tfrac{B_{n_\ell}(z,w)}{\babs{z-w}}dA(w)\\
&\leq \left(2/\delta\right)^{2c'}M\int_{\babs{z-w}<\delta/2}\tfrac{dA(w)}{\babs{z-w}} + 2\delta^{-1}\\
&\leq 2^{2c'}M\delta^{1-2c'} + 2\delta^{-1}.
\end{align*}
}
This shows that $|C_{n_\ell}(z)|\le M|z|^{-1}$ for all $z
\in \calV\setminus(\calZ\cup\{0\})$, where the constant $M$ depends only on $\calV$.
\end{proof}

Finally, the following lemma concludes the proof of Theorem \ref{m1.5}.
\begin{lem} If $R$ does not vanish identically, then $R>0$ everywhere on $\C^*$ and
\begin{equation}\label{ward:dist} \bar{\d}C=R-\Lap V_0-\Lap \log R
\end{equation}
in the sense of distributions on $\C$ and pointwise on $\C^*$.
\end{lem}

\begin{proof}
By Lemma \ref{lem:ward}, we have
\begin{equation}\label{ward:n} \bar{\d} C_n(z) = R_n(z) - \Lap V_0(z) -\Lap \log R_n(z) + o(1)
\end{equation}
where $o(1)\to 0$ as a distribution on $\C$ and uniformly on each compact subset in $\C^*$. For a compact subset $\calV$ of $\C$, we moreover know that $\calV\cap\calZ$ is a finite set and $zC_{n_\ell}(z)\to zC(z)$ boundedly and locally uniformly on $\calV\setminus(\calZ\cup \{0\})$. {For a test function $\phi$ and a small number $\epsilon$, we take $\mathscr{V}$, $\mathscr{V}_{\epsilon}$ such that
$$\mathscr{V}=\supp \phi,\quad \mathscr{V}_{\epsilon}=\mathscr{V}\setminus \bigcup_{z\in \mathscr{Z}\cup \{0\}}D(z,\epsilon).$$
Here, $\mathscr{V}_{\epsilon}$ is compact and the (normalized) area of $\mathscr{V}\setminus\mathscr{V}_{\epsilon}$ is less than $N \epsilon$ for some $N>0$.
Moreover, since $|C_{n_\ell}(z)|\leq M|z|^{-1}$ for all $z$ near $0$ as in the proof of Lemma \ref{C:bound}, we get
$$\int_{D(0,\epsilon)}|C_{n_\ell}-C|\,dA \leq \int_{D(0,\epsilon)} 2M|z|^{-1}\,dA(z) \leq
4M\epsilon.$$
Thus the integral
\begin{align*}
\int_{\mathscr{V}} (C_{n_\ell}-C)\,\phi \,dA
&= \int_{\mathscr{V}_{\epsilon}} (C_{n_\ell}-C)\,\phi\,dA + \int_{\mathscr{V}\setminus \mathscr{V}_{\epsilon}} (C_{n_\ell}-C)\,\phi \,dA
\end{align*}
can be made small since $C_{n_\ell}\to C$ uniformly on $\mathscr{V_{\epsilon}}$ and $zC_{n_\ell}(z)$ is uniformly bounded on $\calV \setminus (\calZ\cup \{0\})$.}
Thus $C_{n_\ell}\to C$ in the sense of distributions on $\C$, which implies $\bar{\d} C_{n_\ell}\to \bar{\d} C$ in the same sense. In view of \eqref{ward:n} and the convergence $R_{n_\ell}\to R$ such that $R_{n_\ell}(z) \leq M|z|^{2c}$ for $|z| \leq 1$ and $R_{n_\ell}(z)$ are uniformly bounded in each compact subset of $\C^*$, we
conclude that the measures $\Lap\log R_{n_\ell}$ converge to $\Lap \log R$.
Passing to the limit as $\ell\to\infty$ we obtain \eqref{ward:dist} in the sense of distributions.

We now prove that $R>0$ on $\C^*$.
For this, we suppose $R(z_0)=0$ for some $z_0 \ne 0$. Then by Lemma \ref{L:zero},
$$R(z)=\babs{z-z_0}^2 L_1(z,z)e^{-V_0(z)}$$
 for some Hermitian-entire function $L_1$. We fix a small disk $D \subset \C^*$ centered at $z_0$ where $L_1(z,z)$ is logarithmically subharmonic (cf. Lemma \ref{logsub}), and
define two measures $\nu$ and $\nu_1$ by
\begin{align*}
d\nu(z)&=\1_D(z)\cdot \Lap_z \log L(z,z) \, dA(z), \\
d\nu_1(z)&=\1_D(z)\cdot \Lap_z \log L_1(z,z)\,dA(z).
\end{align*}
 By Lemma \ref{logsub}, these measures are both positive.

Now note that $\nu=\nu_1+\delta_{z_0}$ and consider the Cauchy transform
$$C^\nu(z)= \int \frac{d\nu(w)}{z-w}.$$
Since $\dbar C^{\nu}= \Lap Q_0 + \Lap\log R$ on $D$,
the Ward's equation \eqref{ward:dist} gives $\bar{\d} (C+C^{\nu})(z)= R(z)$ on $D$, whence
\begin{equation*}
C(z)+C^{\nu}(z)=v(z)\end{equation*} for some smooth function $v$ on $D$. Since $C$ is bounded in $D\setminus\{z_0\}$ by Lemma \ref{C:bound}, $C^{\nu}(z)$ remains bounded as $z\to z_0$. This implies that $\nu(\{z_0\})=0$, which contradicts that $\nu_1$ is positive. The contradiction shows that $R(z_0)>0$.

The positivity of $R$ implies that $\log R$ is real-analytic on $\C^*$, so the right hand side of \eqref{ward:dist} is smooth there. By Weyl's lemma which asserts that every weak solution of Laplace's equation is smooth, $C(z)$ is also smooth on $\C^*$ and hence Ward's equation holds pointwise on $\C^*$.
\end{proof}

\section{Symmetric solutions in the dominant radial case} \label{ud}

We will now prove our first principal result, Theorem \ref{mt2}. Our proof elaborates on arguments from the papers \cite{AK,AS}. We will initially allow for kernels $L$ which
are not necessarily rotationally symmetric.

\smallskip

We start by noting that if $R$ is a non-trivial limiting $1$-point function in Theorem \ref{m1.5}, then Ward's equation can be written in the form
$$\dbar C(z)=R(z)-\Lap_z \log L(z,z),$$
where $L$ is the corresponding limiting $1$-point function, $R(z)=L(z,z)e^{-V_0(z)}$.

We here consider the dominant radial case
$$Q_0(z) = Q_0(\babs{z}),\qquad V_0(z)=Q_0(z)-2c\log|z|,$$
and we
write
$$d\mu_0=e^{-V_0}\, dA.$$

Since the kernel $L$ is Hermitian-entire, we can represent it in
the form
\begin{equation*}L(z,w)=\sum_{j,k=0}^\infty a_{jk}z^j\bar{w}^k,\quad (a_{jk}=\bar{a}_{kj}).
\end{equation*}
We know also that $L$ is nontrivial, i.e., $L(z,z)>0$ for all $z\ne 0$ by Theorem \ref{m1.5}.

Let $L_0$ be the Bergman kernel for the space $L^2_a(\mu_0)$. Since $Q_0$ is radially symmetric, we have the formula
\begin{equation}\label{l0f}L_0(z,w)= \sum_{j=0}^{\infty} \frac{(z\bar{w})^j}{\norm{z^j}^2_{L^2(\mu_0)}}.\end{equation}
We want to prove that $L=L_0$.

\smallskip

Let us rewrite Ward's equation. Below we fix a complex number $z\ne 0$.

The Cauchy transform $C(z)$ is computed as follows:
\begin{align*}
C(z)&=\frac{1}{L(z,z)} \int_{\C} \frac{\babs{L(z,w)}^2}{z-w}e^{-Q_0(w)+2c\log\babs{w}}dA(w) \\
&=\frac{1}{L(z,z)} \sum_{j,k,l,m} a_{jk}\bar{a}_{lm} z^j \bar{z}^{l}\int_{\C} \frac{\bar{w}^k w^m}{z-w} e^{-Q_0(w)+2c\log\babs{w}}dA(w)\\
&=\frac{1}{L(z,z)}\sum_{j,k,l,m} a_{jk} a_{ml} z^j \bar{z}^l \int_{0}^{\infty}
r^{k+m} e^{-Q_0(r)+2c\log r}  dr \int_{0}^{2\pi} \frac{e^{i(m-k)\theta}}{z/r-e^{i\theta}} \frac{d\theta}{\pi}.
\end{align*}
From the fact that
\begin{align*}
\frac{1}{2\pi} \int_{0}^{2\pi} \frac{e^{i(m-k)\theta}}{z/r-e^{i\theta}}d\theta
= \begin{cases}
-\left(z/r\right)^{m-k-1} & \mathrm{if}\ \babs{z}<r, \ m-k\geq 1,\\
 \quad \left(z/r \right)^{m-k-1} & \mathrm{if}\ \babs{z}>r, \ m-k\leq 0, \\
 \quad \ 0 & \mathrm{otherwise},
\end{cases}
\end{align*}
we obtain that (all sums are over all $j,k,l,m$ unless otherwise is specified)
\begin{align*}
C(z) &=  \frac{2}{L(z,z)} \sum_{m\le k} a_{jk} a_{ml} z^j \bar{z}^l  \int_{0}^{\babs{z}}  r^{k+m} \left( \frac{z}{r} \right)^{m-k-1} e^{-Q_0(r)+2c\log r}dr, \\
&-  \frac{2}{L(z,z)} \sum_{m\ge k+1} a_{jk} a_{ml} z^j \bar{z}^l  \int_{|z|}^{\infty}  r^{k+m} \left( \frac{z}{r} \right)^{m-k-1} e^{-Q_0(r)+2c\log r} dr. \end{align*}
This can be rewritten $C(z)=S_1(z)-S_2(z)$ where
\begin{align*}S_1(z)&=\frac 2 {L(z,z)}\sum_{j,k,l,m}a_{jk}a_{ml}z^{j+m-k-1}\bar{z}^l\int_0^{|z|}
r^{2k+1}e^{-Q_0(r)+2c\log r}\, dr\\
&=\sum_{j,k}a_{jk}z^{j-k-1}\int_0^{|z|^2}t^{k}e^{-Q_0(\sqrt{t})+c\log t}\, dt.
\end{align*}
and
\begin{equation}\label{s2}S_2(z)=\frac 1 {L(z,z)}\sum_{m\ge k+1}a_{jk}a_{ml}z^{j+m-k-1}\bar{z}^l\int_0^{\infty}t^k e^{-Q_0(\sqrt{t})+c\log t}\, dt.\end{equation}

It follows that
\begin{align*}\dbar S_1(z)&=z\sum_{j,k} a_{jk}z^{j-k-1}|z|^{2k}e^{-Q_0(z)+c\log|z|^2}\\
&=L(z,z)e^{-V_0(z)}=R(z),\qquad (z\ne 0).
\end{align*}
Hence Ward's equation, which can be written in the form
$$\dbar S_1(z)-\dbar S_2(z)=R(z)-\Lap_z\log L(z,z)$$
is equivalent to just
$$\dbar(S_2(z)-\d\log L(z,z))=0.$$
Now,
$$\d\log L(z,z)=\frac 1 {L(z,z)}\sum_{m,l}ma_{ml}z^{m-1}\bar{z}^l,$$
so by \eqref{s2}, the distributional Ward's equation is equivalent to that the function
\begin{equation}\label{holo}\frac 1 {L(z,z)} \sum_{m,l}a_{ml}z^{m-1}\bar{z}^l(\sum_{j=0}^\infty\sum_{k=0}^{m-1}a_{jk}z^{j-k}\|z^k\|_{L^2(\mu_0)}^2-m)\end{equation}
has $\dbar$-derivative $0$ in the sense of distributions on $\C$. By Weyl's lemma, this implies that there exists some
entire function $G(z)$ such that
\begin{equation*}\sum_{m,l}a_{ml}z^{m-1}\bar{z}^l(\sum_{j=0}^\infty\sum_{k=0}^{m-1}a_{jk}z^{j-k}\|z^k\|_{L^2(\mu_0)}^2-m)
=G(z)\sum_{m,l}a_{ml}z^{m}\bar{z}^l.\end{equation*}

The Taylor coefficients $g_j$ in $G(z)=\sum g_jz^j$ depend on $a_{jk}$ in a complicated
way. Comparing coefficients of $\bar{z}^l$ in \eqref{holo}, we deduce for each $l\in\Z_+$ the identity
\begin{equation}\label{bod2}\sum_{m=1}^\infty a_{ml} (\sum_{j=0}^{\infty}
\sum_{k=0}^{m-1}a_{jk}z^{j+m-1-k}\|z^k\|_{L^2(\mu_0)}^2-mz^{m-1})
=G(z)\sum_{j=0}^\infty a_{jl}z^j.\end{equation}
Thus, for example, the constant term $g_0=G(0)$ obeys infinitely many relations
$$g_0a_{0l}=a_{1l}(a_{00}\|1\|^2_{L^2(\mu_0)}-1) +
\sum_{m=2}^\infty a_{ml}\,a_{0,m-1}\|z^{m-1}\|^2_{L^2(\mu_0)},\quad l=0,1,2,\ldots.$$

Due to the formidable appearance of these relations, we now abandon the quest
for the most general limiting kernel $L$ and restrict our attention to symmetric ones.

We thus assume that $L(z,w)=\sum a_{jk}z^j\bar{w}^k$ is rotationally symmetric, i.e., we assume that
$a_{jk}=a_j\delta_{jk}$ for some numbers $a_j$.
In this case, the system \eqref{bod2} becomes: for all $m\ge 1$
\begin{equation*}a_{m}z^{m-1}(\sum_{j=0}^{m-1}a_j\|z^j\|_{L^2(\mu_0)}^2-m)=
G(z)a_mz^m.\end{equation*}
By considering the smallest $m\ge 1$ such that $a_m\ne 0$ we conclude that $G\equiv 0$, i.e.,
\begin{equation}\label{bod4}a_{m}(\sum_{j=0}^{m-1}a_j\|z^j\|_{L^2(\mu_0)}^2-m)=0,\qquad m=1,2,\ldots.
\end{equation}

Before proving Theorem 5, recall that $L$ satisfies the mass-one inequality \eqref{mass}. The mass-one inequality can be written as
\begin{equation}{\label{massco}}\sum_{j=0}^\infty a_j^2|z|^{2j} \|z^j\|^2_{L^2(\mu_0)} \leq \sum_{j=0}^{\infty} a_j|z|^{2j}.\end{equation}

\begin{lem} \label{nupp} Let $L$ be a nontrivial rotationally invariant limiting kernel. Then
$$a_{j}= 1/\|z^j\|_{L^2(\mu_0)}^2,\qquad j=0,1,2,\ldots.$$
\end{lem}
\begin{proof}
We use \eqref{bod4}, \eqref{massco}, and induction on $m\ge 0$.

Put $N_m = \inf\{j\,;\, j\geq m,\, a_j\ne 0\}$ and note that, since $L(z,z)\sim \Lap Q_0(z)e^{V_0(z)}$ as $z\to\infty$
(by Theorem \ref{m1.2}), we have $N_m<\infty$ for each $m$.

Assume that $N_0 > 0$. Then $a_j=0$ for all $j\ge 0$ with $j< N_0$. Since $a_{N_0}\ne 0$, $\sum_{j=0}^{N_0-1}a_j\|z^j\|^2_{L^2(\mu_0)}=N_0$ by \eqref{bod4} which gives a contradiction. Thus $N_0 = 0$, and $0<a_0 \leq 1/\|1\|^2_{L^2(\mu_0)}$ by \eqref{massco}.

Suppose next that $N_1>1$. As above, $a_j=0$ when $j\ge 1$ and $j \leq N_1 -1$. Hence
$\sum_{j=0}^{N_1-1}a_j\|z^j\|^2_{L^2(\mu_0)}=a_0\|1\|^2_{L^2(\mu_0)}\leq 1<N_1$, which contradicts \eqref{bod4} with $m=N_1$. Hence $N_1=1$ and $a_1\ne 0$. By \eqref{bod4}, $a_0=1/\|1\|^2_{L^2(\mu_0)}$ and by \eqref{massco}, $0<a_1\leq 1/\|z\|^2_{L^2(\mu_0)}$.

Similarly, for $m\geq 2$ we have $a_{m-1}=1/\|z^{m-1}\|^2_{L^2(\mu_0)}$ and $0<a_m \leq 1/\|z^m\|_{L^2(\mu_0)}$, finishing the induction step.
\end{proof}

A comparison of the formula \eqref{l0f} for the kernel $L_0$ with the above lemma shows that $L=L_0$, finishing our proof of
Theorem \ref{mt2}.  q.e.d.

\section{Universality at radial type singularities} \label{ginen}

We now prove Theorem \ref{mt22}.

Suppose that the potential $V_n$ is of the form
$$V_n(\zeta)=Q(\zeta)-\tfrac {2c}n\log|\zeta|,\qquad Q=Q_r+h$$
where $Q_r$ is radially
symmetric and $Q_r(z)=|z|^{2\lambda} + O(|z|^{2\lambda+\epsilon})$ as $z\to 0$ for some $\epsilon>0$ and where $h$ is a harmonic polynomial with
$h(0)=0$. We assume that the degree of $h$ is at most $d$, where, without loss of generality, $d\ge 2\lambda$.

 We suppose in addition
 that $0\in\Int S$, that $S$ is connected, and that $\d\Pc S$ is an everywhere regular Jordan curve.

Now write $h=\re H$ where $H$ is the holomorphic polynomial
$$H(\zeta)=Q(0)+2\d Q(0)\cdot \zeta+\cdots+\tfrac 2 {d!}\d^dQ(0)\cdot \zeta^d.$$

The crucial property that we shall need from $Q_r$ is that
$$\zeta\d Q_r(\zeta)\in\R$$
for all $\zeta$. This is satisfied since $Q_r$ is radially symmetric.

Let $p_{n,j}$ be the $j$:th orthonormal polynomial with respect to $e^{-nV_n}$ and rescale about $0$ by letting
$$q_{n,j}(z)=r_n^{1+c}p_{n,j}(r_nz),\quad \tilde{V}_n(z)=nQ(r_nz)-2c\log|z|.$$
With $\tilde{H}_n(z)=nH(r_nz)$ we can then define a ``rescaled holomorphic kernel'' by
$$L_n(z,w)=\sum_{j=0}^{n-1}q_{n,j}(z)\bar{q}_{n,j}(w)
e^{-(\tilde{H}_n(z)+\bar{\tilde{H}}_n(w))/2}.$$
By a straightforward extension of normal families argument (see Section \ref{gener}) we obtain easily that each subsequence of $L_n$ has a further subsequence converging locally uniformly to a limiting holomorphic kernel $L$. (The terms in $\tilde{H}_n$ of degree $>2\lambda$ are negligible on compact sets.)

We want to prove the asymptotic rotation invariance
$$\d_\theta L_n(z,z)\to 0$$
as $n\to\infty$ along any subsequence. To prove this we observe that
$$\d_\theta L_n(z,z)=iz\d_z L_n(z,z)+(-i\bar{z})\dbar_zL_n(z,z)=-2\im(z\d_zL_n(z,z)).$$
Write $\tilde{h}_n=\re \tilde{H}_n$. Then
\begin{equation}\label{zdLn}
z\d_zL_n(z,z)=e^{-\tilde{h}_n(z)}\sum_{j=0}^{n-1}(z\d q_{n,j}(z)-zq_{n,j}(z)\d \tilde{h}_n(z))\bar{q}_{n,j}(z).
\end{equation}
Since $\d \tilde{h}_n$ is a holomorphic polynomial of degree at most $d-1$, $zq_{n,j}\d \tilde{h}_n$ is a polynomial of degree at most $j+d$. Hence it can be expressed as a linear combination of the polynomials $q_{n,l}$:
$$zq_{n,j}\d \tilde{h}_n(z)=\sum_{l=0}^{j+d}\langle zq_{n,j}\d \tilde{h}_n,q_{n,l}\rangle \, q_{n,l}(z)$$
where $\langle f,g\rangle=\int f\bar{g}e^{-\tilde{V}_n}\, dA$.

Now set $\tilde{Q}_r=nQ_r(r_nz)=|z|^{2\lambda}+\cdots$. An integration by parts gives
\begin{align*}\langle zq_{n,j}\d \tilde{h}_n,q_{n,l}\rangle &=
-\int zq_{n,j}(z)\bar{q}_{n,l}(z)\d(e^{-\tilde{h}_n(z)})e^{-\tilde{Q}_r+2c\log|z|}\, dA(z)\\
&=\langle q_{n,j},q_{n,l}\rangle+\langle z\d q_{n,j},q_{n,l}\rangle\\
&\quad+
\int zq_{n,j}(z)\bar{q}_{n,l}(z)(-nr_n\d Q_r(r_nz)+\frac c z)e^{-\tilde{V}_n(z)}\, dA(z).
\end{align*}
Let us define
\begin{equation}\label{ajl}\alpha_{j,l}=\int q_{n,j}(z)\bar{q}_{n,l}(z)nr_nz\d Q_r(r_nz)e^{-\tilde{V}_n(z)}\, dA(z).\end{equation}
We then have the identity
$$\langle zq_{n,j}\d \tilde{h}_n,q_{n,l}\rangle=(1+c)\langle q_{n,j},q_{n,l}\rangle+
\langle z\d q_{n,j},q_{n,l}\rangle-\alpha_{j,l}
$$
and from \eqref{zdLn},
\begin{align*}
z\d_z L_n(z,z)  &= e^{-\tilde{h}_n(z)} \sum_{j=0}^{n-1}\bar{q}_{n,j}(z)\Big(\sum_{l=0}^{j+d}\langle z\d q_{n,j} - zq_{n,j}\d \tilde{h}_n, q_{n,l}\rangle q_{n,l} \Big) \\
&=e^{-\tilde{h}_n(z)}(-(1+c)\sum_{j=0}^{n-1}|q_{n,j}(z)|^2+\sum_{j=0}^{n-1}
\sum_{l=0}^{j+d}\alpha_{j,l}q_{n,l}(z)\bar{q}_{n,j}(z)).
\end{align*}
Since $z\d Q_r(z)=\lambda|z|^{2\lambda}+\cdots$ is real-valued, we have $\alpha_{j,l}=\bar{\alpha}_{l,j}$ and hence
$$\alpha_{j,l}q_{n,l}(z)\bar{q}_{n,j}(z)+\alpha_{l,j}q_{n,j}(z)\bar{q}_{n,l}(z)\in \R$$
for all $l,j,z$. Therefore,
\begin{equation}\label{iml}\im(z\d_zL_n(z,z))=e^{-\tilde{h}_n(z)}
\sum_{j=n-d}^{n-1}
\sum_{l=n}^{d+j}
\im(\alpha_{j,l}q_{n,l}(z)\bar{q}_{n,j}(z)).
\end{equation}

In order to estimate the sum in the right hand side of \eqref{iml}, we need a good control of the polynomials $q_{n,j}$ with $|j-n|\le d$.
For this purpose, we will employ recent results on asymptotics for the orthonormal polynomials $p_{n,j}$
with respect to $e^{-nV_n}$, to obtain the following result.

\begin{lem} \label{Vk} Let $K$ be any fixed compact subset of
$\Int\Pc S$. Then for any given integer $\kappa\ge 0$ and any $j=j(n)$ with $|j-n|\le d$ we have
$$\int_K|p_{n,j}|^2e^{-nV_n}\, dA=O(n^{-\kappa-1}),\qquad n\to\infty.$$
\end{lem}

\begin{proof} We will use one of the main results from the paper \cite{HW} \footnote{In the
special case when $Q=|\zeta|^2+\re t\zeta$ we can alternatively use the strong asymptotics proved in the papers \cite{LY,WW}}, which implies that if
$V_n$ is any potential of the type indicated above,
 then for any fixed $\kappa\ge 0$ we have
\begin{equation}\label{fjun}\int|p_{n,j}-\chi F_{n,j}^{(\kappa)}|^2e^{-nV_n}\, dA=O(n^{-\kappa-1}),\quad |j-n|\le d.\end{equation}
Here $\chi$ is a smooth function which equals to $1$ in a small tubular neighbourhood of the outer boundary of $S$ while $\chi=0$ in the complement of a slightly larger tubular neighbourhood, which is still small enough that $\chi=0$ near all the logarithmic singularities occurring in the potential.

Following \cite{HW} we write $\calQ$ for (the analytic continuation of) the bounded holomorphic function on $\C\setminus \Pc S$ satisfying $\re\calQ=Q$ on $\d\Pc S$. (We fix $\calQ$ uniquely by requiring that the imaginary part vanishes at infinity.)

The function $F_{n,j}^{(\kappa)}$ is of the form
$$F_{n,j}^{(\kappa)}(\zeta)=n^{1/4}\sqrt{\phi'(\zeta)}\phi(\zeta)^{j} \,e^{n\calQ(\zeta)/2}\sum_{p=0}^\kappa n^{-p}\calB_p(\zeta)$$
where $\phi$ is the analytic continuation of a univalent map $\hat{\C}\setminus \Pc S\to \hat{\C}\setminus \D$ and $\calB_p$ are certain holomorphic functions, bounded on the support of $\chi$.

It is clear that the asymptotic in \eqref{fjun} implies the statement in the lemma.
\end{proof}

\begin{rem} In \cite{HW} there is an assumption that the potential $Q$ be everywhere $C^2$-smooth in the plane. It is however easy to see that
the results hold under the assumptions in the present paper. Indeed, the details of the potential in a given compact subset of $\Int S$ are
completely irrelevant in \cite{HW} (as long as it is admissible in the sense of \cite{ST}). Similarly, the details
of the potential in the exterior, away from the boundary, are not relevant. The only point where any kind of higher regularity is used
is in the form of real-analyticity near the boundary.
\end{rem}

\begin{lem} The integral $\alpha_{j,l}$ in \eqref{ajl} satisfies
$\alpha_{j,l}=O(n)$ when $|j-n|\le d$, $|l-n|\le d$,
and $n\to\infty$.
\end{lem}

\begin{proof} By a change of variables we have
\begin{equation}\label{bo}\alpha_{j,l}=n\int z\d Q_r(z) p_{n,j}(z)\bar{p}_{n,l}(z) \,e^{-nV_n(z)}\, dA(z).\end{equation}
Since the $|p_{n,j}|^2 e^{-nV_n}$ are negligible outside the support of $\chi$ (see remark below) we can restrict to some small neighbourhood of the droplet. We can then bound the integral in \eqref{bo}
by $C\int|p_{n,j}||p_{n,l}|e^{-nV_n}$, so by the Cauchy-Schwarz inequality it follows that $|\alpha_{j,l}|\le Cn$.
\end{proof}

\begin{rem} To make the proof above precise, we here provide an exterior estimate of $|p_{n,j}|^2e^{-nV_n}$ on $\C\setminus(\supp \chi\cup \Pc S)$, i.e., outside of a neighborhood of $\Pc S$. Indeed, a modification of the pointwise asymptotics in Theorem 1.3.5 from \cite{HW} gives the asymptotics
$$|p_{n,j}|^2 e^{-nV_n} = \sqrt{n}|\phi'||\phi|^{2j}e^{-n(Q-\re\calQ)}|\calB_0+O(n^{-1})|^{2}$$ which is uniform for $z$ outside a neighborhood of $S$.
Moreover, according to the terminology in the paper \cite{HW},
$\re \calQ +\log |\phi|^2 = \breve{Q}$, which gives a very fast decay of $|p_{n,j}|^2 e^{-nV_n}$ outside an arbitrarily small neighborhood of $\Pc S$.
\end{rem}

Now let $K$ be a given compact subset of $\C^*$. We shall estimate the integral
$$J_n:=\int_K|\d_\theta L_n(z,z)|\, dA(z).$$
By \eqref{iml} and the last lemma,
$$J_n\le Cn\sum_{j=n-d}^{n-1}\sum_{l=n}^{d+j}\int_K |q_{n,l}||q_{n,j}|e^{-\tilde{V}_n}\, dA.$$
Hence by Lemma \ref{Vk}, we have $J_n=O(n^{-\kappa})$ for any given $\kappa>0$. (The compact set $K$ corresponds under the rescaling to the dilated set $r_nK$ which is surely bounded away from the outer boundary when $n$ is large enough.)

It follows that if $L=\lim L_{n_l}$ is a limiting holomorphic kernel, then for each compact set $K\subset\C$,
$\int_K|\d_\theta L(z,z)|\, dA(z)=0$. This is only possible if $\d_\theta L(z,z)=0$ identically, i.e., $L$ is
rotationally invariant.

The proof of Theorem \ref{mt22} is complete. q.e.d.

\section{Insertion as a balayage operation} \label{sec5}

In this section, we will discuss the effect of insertion of a point charge, by comparing the
first intensity of the process $\{\zeta_j\}_1^n$ with respect to the inserted point charge $c$
to the intensity of the corresponding process with $c=0$.
Our discussion elaborates on
the concluding remarks
from \cite[Section 7.7]{AHM2}.

\smallskip

We will now exploit the freedom of choosing the perturbation $\pert$ in the potential
\begin{equation*}V_n=Q-\tfrac {2c}n\ell-\tfrac 1 n \pert,\end{equation*}
where we write $\ell(\zeta)=\log|\zeta|$. Here, we take $u$ to be smooth and harmonic. One natural choice is to put $\pert=0$, i.e., to set
\begin{equation}\label{no1}V_n=Q-\tfrac {2c} n\ell\end{equation}
We will call \eqref{no1} the ``pure $\log$-normalization'' of $V_n$.
Another interesting choice, the ``Green's form'', is to set
\begin{equation}\label{pot2}V_n=Q+\tfrac {2c}nG\end{equation}
where $G$ is the Green's function for $S$ with pole at $0$,
\begin{equation}\label{gu}G(\zeta)=\log\tfrac 1{|\zeta|}-\pert(\zeta),\end{equation}
$\pert$ is harmonic in $\Int S$ and $G=0$ on the boundary $\d S$.

\smallskip

In order to apply the theory of \cite{AHM2,AM} we need to impose a few conditions on the geometry and
topology of the droplet $S$.
First of all, we recall our assumption that $Q$ be real-analytic in a neighbourhood of $S$ with exception at the origin. (See \eqref{lam_lim}.) We will now
also impose the condition that $\Lap Q$ is strictly positive in a neighbourhood of the boundary $\d S$.
It is well-known that this guarantees that $S$ has finitely many components, and that
$\d S$ is a finite union of real-analytic curves, possibly having finitely many singular points
of known types (cusps or double points). (Cf. \cite{AKMW,LM} and references.) We shall assume that $S$ is connected, and that the boundary
is everywhere smooth. In this case, the Green's function $G(\zeta)$ can be extended to a smooth function on $\C^*$ with $G\equiv 0$ near $\infty$. We fix such an extension and insist on calling it $G$.

\smallskip

Let $\bfR_n$ and $\tilde{\bfR}_n$ be the $1$-point functions associated, respectively, with the
external potentials $Q$ and $V_n$.
We shall measure the effect of the insertion by studying the difference
$$\rho_n=\tilde{\bfR}_n-\bfR_n,$$
where we use the convention that $\bfR_n$ is identified with the measure $\bfR_n\, dA$ (and likewise
for $\tilde{\bfR}_n$).

Note that the asymptotics of $\rho_n(f)$ can be nontrivial only if the support of $f$ contains
either the point $0$, or some portion of the boundary of $S$. Indeed, if we exclude neighbourhoods of $0$ and
of the boundary, then on the rest of the plane
both $\bfR_n$ and $\tilde{\bfR}_n$ are very close
to $n\Lap Q\cdot \1_S$, in the sense that the difference is negligible, cf. \cite{AKM,AS2}.

\smallskip

In order to simplify the following discussion, we assume now that the potential $Q$ has a dominant
radial bulk singularity at $0$ of Mittag-Leffler type, and that the coefficient $\tau_0=1$. That is, we assume according to the canonical decomposition \eqref{candec} that
\begin{equation*}Q_0=|z|^{2\lambda},\quad V_0=|z|^{2\lambda}-2c\log|z|.\end{equation*}
In addition we will assume that \textit{each limiting kernel is rotationally symmetric}, that is (by Theorem \ref{mt2})
each limiting kernel is equal to $L(z,w)=\lambda\cdot E_{1/\lambda,(1+c)/\lambda}(z\bar{w})$.
Finally, we assume that there are no other singular points in the bulk, except the one at $0$.

\smallskip

In the rest of this section, we assume that all of the above conditions are satisfied.

For a domain $\Omega$ bounded by a Jordan curve $\gamma$ in $\C$,
the harmonic measure with respect to $\Omega$ evaluated at a point $z\in\Omega$ is the hitting probability distribution on $\gamma$ of
a Brownian motion starting at $z$. We refer to \cite{GM} for other equivalent definitions of the harmonic measure.

\begin{mth} \label{thm1} Let $\tilde{\bfR}_n$ be the $1$-point functions associated with the potential
\begin{enumerate}[label=(\roman*)]
\item\label{fi} $V_n=Q+2(c/n)G$. Then, in the weak sense of measures
$$\rho_n\to c(\omega_0 -\delta_0),$$
Here $\delta_0$ is the Dirac measure at $0$ and $\omega_0$ the harmonic measure with respect to $S$, evaluated at
$0$.
\item\label{se}
$V_n=Q-(2c/n)\ell$.
Then
$$\rho_n\to c(\omega_\infty -\delta_0),$$
where $\omega_\infty$ is the harmonic measure with respect to $\C\setminus S$, evaluated at $\infty$.
\end{enumerate}
\end{mth}

The theorem shows that the insertion of a point mass, corresponding to different natural boundary conditions,
gives rise to different kinds of balayage operations, see Figure \ref{bala}, cf. \cite{GM,ST} for the basic facts about balayages.

\begin{figure}[ht]
\begin{center}
\includegraphics[width=.4\textwidth]{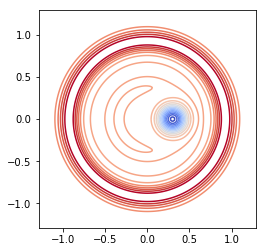}
\qquad
\includegraphics[width=.4\textwidth]{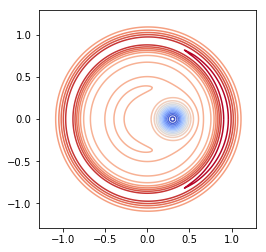}
\end{center}
\caption{Some level curves of $\rho_n$ with respect to Ginibre potential $V_n=|\zeta|^2-2c\log|\zeta-a|/n$
when $n=40$, $c=1$, $a=0.3$ with pure-log normalization and Green's normalization, respectively. (Blue means negative and red means positive.)} \label{bala}
\end{figure}

\begin{lem} \label{vott} Suppose that the test-function $f$ has compact support in the interior of $S$. Then
$$\lim_{n\to\infty}\rho_n(f)=-cf(0).$$
In other words, $\rho_n\to -c\delta_0$ in the interior of the droplet.
\end{lem}

\begin{proof} Rescaling about $0$ on the scale $n^{-1/2\lambda}$, we obtain the $1$-point intensities $R_n$ and $\tilde{R}_n$ which,
by Theorem \ref{mt2} have known locally uniform limits
$R$ and $\tilde{R}$. We shall show that the total mass of the measure $R-\tilde{R}$ is
\begin{equation}\label{box}(R-\tilde{R})(\C)=c.\end{equation}
In view of well-known bulk estimates of $\bfR_n$ and $\tilde{\bfR}_n$ away from a microscopic neighbourhood of
$0$, this will prove the lemma. Before proceeding with the proof, we pause to recall some details about the estimates that come into play here.

Consider the annular regions
\begin{equation*}
A_n:\quad Mr_n \leq |\zeta| \leq r_n (\log n)^{\frac 1{\lambda}},\qquad
B_n:\quad r_n (\log n)^{\frac{1}{\lambda}} \leq |\zeta| \leq d.
\end{equation*}
where $M$ is large
and $d$ is some small constant so that $B_n$ lies inside the droplet. We claim that the total variation of $\rho_n$ on $A_n\cup B_n$ can be made less than any given $\epsilon>0$ by choosing large $M$ and small $d$.
Indeed, for $\zeta \in A_n$, the asymptotics from \cite[Section 3]{AS2} shows that there exist some constants $c$, $C$ such that
$$|\tilde{R}_n(z)-\Lap Q_0(z)| \leq C|z|^{2\lambda-2} e^{-c|z|^{2\lambda}}\!,\quad |R_n(z)-\Lap Q_0(z)| \leq C|z|^{2\lambda-2} e^{-c|z|^{2\lambda}}$$ for all large $n$. Thus
\begin{align*}
\int_{M\leq |z|\leq (\log n)^{\frac{1}{\lambda}}}|\tilde{R}_n(z) - R_n (z)| dA(z) \leq 2C\int_{M\leq |z|\leq (\log n)^{\frac{1}{\lambda}}}|z|^{2\lambda-2} e^{-c|z|^{2\lambda}} dA(z),
\end{align*} where the right hand side can be made small by taking $M$ large enough.

By a slight adaptation of the H\"{o}rmander estimate in \cite[Section 5]{AKM}, we obtain the following estimate.

\begin{lem}
There exist some constants $c$, $C>0$ such that for $\zeta\in B_n$
\begin{equation}\label{voxd}
\begin{cases}
|\tilde{\bfR}_n(\zeta)-n\Lap Q(\zeta)| &\leq C(1+n^{\max(1,\frac{1}{\lambda})}\cdot e^{-c(\log n)^2})\cr
&\cr
|\bfR_n(\zeta)-n\Lap Q(\zeta)| &\leq C(1+n^{\max(1,\frac{1}{\lambda})}\cdot e^{-c(
\log n)^2})\cr
\end{cases}
\end{equation}
for all large $n$.
\end{lem}
\begin{proof}
Using the pointwise-$L^2$ estimate in \cite[Section 3]{AKM}, we obtain the elementary estimate for $\zeta \in B_n$
\begin{equation}\label{elbd}
\bfR_n(\zeta) \leq C n \Lap Q(\zeta)
\end{equation}
for some constant $C>0$.

For a point $\zeta\in B_n$ and a small enough constant $\alpha$, let $\epsilon_n(\zeta)=\alpha\log n/\sqrt{n\Lap Q(\zeta)}$. Consider a smooth function $\chi_n$ such that $\chi_n = 1$ on $D(\zeta,\epsilon_n/2)$, $\chi_n =0$ outside of $D(\zeta,\epsilon_n)$, and $\|\dbar \chi_n\|_{L^2}\le C$. With a suitable choice of $\alpha$, $D(\zeta, \epsilon_n)$ does not contain the origin. In order to use the estimate in \cite[Section 5]{AKM}, we write $Q(\eta,\omega)$ for the Hermitian-analytic extension of $Q$ satisfying $Q(\eta,\eta)=Q(\eta)$ in a neighborhood of $\zeta$ and
define the approximate kernel $\bfK_n^{\#}$ by
$$\bfK^{\#}_n(\eta,\omega) = n\d_\eta\dbar_\omega Q(\eta,\omega)\cdot e^{nQ(\eta,\omega)}e^{-nQ(\eta)/2-nQ(\omega)/2}.$$
Writing $\bfK_{n,\eta}(\omega)=\bfK_n(\omega,\eta)$,
we define the operator $\Pi_n$ by
$$\Pi_n[f](\eta)=\int f(\zeta)\, \bar{\bfK}_n(\zeta) \, dA(\zeta).$$
By adapting \cite[Theorem 5.1]{AKM}, we readily obtain
$$|\bfK_n(\zeta,\zeta) - \Pi_n[\chi_n\bfK_{n,\zeta}^{\#}](\zeta)| \le CM_n(\zeta)\sqrt{\bfK_{n}(\zeta,\zeta)}$$
for some constant $C$ where $$M_n(\zeta)=\frac{1}{\sqrt{n\Lap Q(\zeta)}}+\frac{1}{\epsilon_n}e^{-c'n\Lap Q(\zeta)\epsilon_n^2}.$$
It follows from \eqref{elbd} that
\begin{align*}
|\bfK_n(\zeta,\zeta) - \Pi_n[\chi_n\bfK_{n,\zeta}^{\#}](\zeta)| & \le C_1 + C_2 \,n \Lap Q(\zeta) \, e^{-c(\log n)^2}.
\end{align*}
If $\lambda<1$, then $\Lap Q(\zeta) \leq C\,(r_n (\log n)^{\frac1 {\lambda}})^{2\lambda-2}\le C\,n^{-1+\frac1{\lambda}}$. Otherwise, $\Lap Q(\zeta)$ is bounded in $B_n$. Thus we have
\begin{align*}
|\bfK_n(\zeta,\zeta) - \Pi_n[\chi_n\bfK_{n,\zeta}^{\#}](\zeta)| & \le C_1 + C_2 \,n^{\max(\frac1{\lambda},1)}  \, e^{-c(\log n)^2}.
\end{align*}
On the other hand, following the argument in \cite[Section 5.2]{AKM} leads to the estimate
\begin{align*}
|\Pi_n[\chi_n \bfK_{n,\zeta}^{\#}]-\bfK_{n}^{\#}(\zeta,\zeta)|\le Cn\Lap Q(\zeta)\,e^{-c(\log n)^2}\le Cn^{\max(\frac1{\lambda},1)}e^{-c(\log n)^2}.
\end{align*}
Combining the above estimates, we obtain the estimate we want.

For the potential $V_n = Q - \frac{2c}{n}\ell - \frac{2}{n}u$, we use a Hermitian-analytic extension in a small neighborhood of $\zeta \in B_n$
$$V_n(\eta, \omega) = Q(\eta, \omega) - \tfrac{c}{n}(\log \eta + \overline{\log \omega}) - \tfrac{2}{n}u(\eta,\omega),$$
where $u(\eta,\omega)$ is a polarization of $u$, to define the approximate kernel $\tilde{\bfK}_n^\#$ by
$$\tilde{\bfK}_n^\#(\eta,\omega) = n \d_{\eta}\bar{\d}_{\omega}V_n(\eta,\omega) \cdot e^{n(V_n(\eta,\omega) - V_n(\eta)/2 - V_n(\omega)/2)}.$$
Here, $\tilde{\bfK}_n^\#(\zeta,\zeta) = n\Lap Q(\zeta) $ since $u$ is harmonic. The argument using approximate Bergman projection in [6, Section 5] considers a local behavior in the ball of radius $\epsilon_n$ away from the origin, so that the same argument applies for  $\tilde{\bfR}_n$.
\end{proof}

The estimates in \eqref{voxd} show that for $\zeta\in B_n$
$$|\tilde{\bfR}_n(\zeta)-\bfR_n(\zeta)| \leq C\left(1+n^{\max(1,\frac{1}{\lambda})}e^{-c(\log n)^2}\right)$$
is bounded, and the total variation of $\rho_n$ over $B_n$ can be made as small as we want by choosing a small $d$. On the other hand, away from a small neighborhood of $0$, the asymptotic formula (cf. \cite{AHM2,AM})
\begin{align*}
\tilde{\bfR}_n &= n\Lap Q+\frac{1}{2}\Lap\log\Lap Q + O(n^{-1})
\end{align*}
where $O(n^{-1})$ is uniform on each compact subset in $\Int S\setminus \{0\}$ can be used. We have thus reduced our problem to proving the
identity in \eqref{box}.

To show \eqref{box} we recall by Theorem \ref{m1.2} that both $R(z)$ and $\tilde{R}(z)$ approach $\Lap Q_0(z)$ quickly
as $z\to\infty$. In particular, the integrals
$$\int_\C(\Lap Q_0-\tilde{R})\, dA\qquad \int_\C(\Lap Q_0-R)\, dA$$
are convergent; if we can compute them,
we can obtain \eqref{box} by a simple subtraction.

Hence it suffices to show that
\begin{equation}\label{sts}\int_\C(\Lap Q_0-\tilde{R})\, dA=c+\frac {1-\lambda}{2}.\end{equation}

But
$\Lap Q_0=\lambda^2|z|^{2\lambda-2},$
and the limiting $1$-point function $\tilde{R}$ equals to
\begin{align*}\tilde{R}(z)&=\lambda E_{1/\lambda,(1+c)/\lambda}(|z|^2)e^{-V_0(z)}=\lambda \sum_{j=0}^\infty \frac{|z|^{2j}}{\Gamma(\frac {j+1+c}{\lambda})}|z|^{2c}e^{-|z|^{2\lambda}}.
\end{align*}
Our problem is thus to evaluate the integral
$$I_{c,\lambda}=\lambda\int_\C(\lambda|z|^{2\lambda-2}-\sum_{j=0}^\infty \frac{|z|^{2j}}{\Gamma(\frac {j+1+c}\lambda)}|z|^{2c}e^{-|z|^{2\lambda}})\, dA(z),$$
which is the same as
$$I_{c,\lambda}=\int_0^\infty(\lambda-e^{-t}\sum_{j=0}^\infty \frac {t^{\frac{j+c+1}{\lambda}-1}}{\Gamma(\frac {j+1+c}{\lambda})})\, dt.$$

We first consider the case when $\lambda$ is a rational number. The computation of the integral can be reduced to the case $\lambda=1$ by the following observation.

A rational number $\lambda$ with $ \lambda > 0$ is expressed as the fraction $p/q$ of two integers $p,\ q\ge 1$. Letting $j=mp + l$ for some integers $m, l$ with $0\leq l \leq p-1$, we have
$$I_{c,\lambda}=\sum_{l=0}^{p-1}\int_0^\infty\left[\frac{1}{q}-e^{-t}\sum_{m=0}^\infty\frac {t^{mq+\frac{q(l+c+1)}{p}-1}}{\Gamma(mq+\frac {q(l+1+c)}{p})}\right]\, dt.$$

Now we note that
$$
E_{1,c}(t) = E_{q,c}(t^q) + tE_{q,c+1}(t^q) + \cdots + t^{q-1}E_{q,c+q-1}(t^q).$$
Writing
$$F_{s}(t)=e^{-t}t^{c+s}E_{q,c+1+s}(t^q),\quad 0\leq s \leq q-1,$$
we have
$$e^{-t}t^{c}E_{1,c+1}(t) = F_0(t) + \cdots + F_{q-1}(t)$$ and
$F_s = F_{s-1} - F'_{s}$ for all $1\leq s\leq q-1$. This implies that
\begin{equation}\label{EFrel}e^{-t}t^{c}E_{1,c+1}(t) = q F_0(t) - \sum_{s=1}^{q-1} (q-s)F'_{s}(t).\end{equation}
Using the asymptotic formula (cf. \cite[eq. (4.4.18)]{GKMR})
$$E_{q,c}(t^q) = \frac{1}{q} \sum_{|2\pi n|<\frac{3\pi q}{4}}(te^{2\pi i n/q})^{1-c}e^{t e^{2\pi i n/q}}+O(t^{-q}),\quad t\to \infty,$$
we now obtain $F_{s}(\infty)=\frac{1}{q}$ for $0\leq s \leq q-1$. In view of equation \eqref{EFrel}, we thus obtain for $c>-1$
$$I_{c,1}=\int_{0}^{\infty} [1-e^{-t}t^cE_{1,c+1}(t)]\, dt =  \int_{0}^{\infty}\left[1-q\,e^{-t}t^cE_{q,c+1}(t^q)\right]dt + \tfrac{q-1}{2}.
$$
If we temporarily accept that $I_{c,1}=c$ for $c>-1$, we now deduce that
$$\int_{0}^{\infty}\left[\frac{1}{q}-e^{-t}t^cE_{q,c+1}(t^q)\right]dt = \frac{1}{q}(c-\frac{q-1}{2})$$
and hence
$$I_{c,\lambda}=\frac{1}{q}\sum_{l=0}^{p-1} (\frac{q(l+c+1)}{p}-1 - \frac{q-1}{2})=c+\frac{1-\lambda}{2}.$$
The identity \eqref{sts} now follows by the continuity in $\lambda$ for all real numbers $\lambda>0$, provided that we can show it for $\lambda=1$.

It remains to verify that $I_{c,1}=c$, where
$$I_{c,1}=\int_0^\infty(1-t^ce^{-t}E_{1,c+1}(t))\, dt.$$
Our verification that $I_{c,1}=c$ is somewhat lengthy, and works in fact for complex $c$ with $\re c>-1$. An alternative short proof in the case when $c$ is a positive integer
is given in the remark below.

We first assume that $\re c>0$. To show that $I_{c,1}=c$, it is convenient to call on the lower and upper incomplete gamma functions
$$\gamma(c,t)=\int_0^ts^{c-1}e^{-s}\, ds,\qquad \Gamma(c,t)=\int_t^\infty s^{c-1}e^{-s}\, ds,\qquad
(\Gamma(c)=\gamma(c,t)+\Gamma(c,t)).$$
We claim that (for $\re c>0$ and $t>0$)
\begin{equation}\label{cla}t^ce^{-t}E_{1,1+c}(t)=\frac {\gamma(c,t)}{\Gamma(c)}.\end{equation}
To see this, we integrate the right hand side by parts,
\begin{equation}\label{j0}\frac1 {\Gamma(c)}\int_0^ts^{c-1}e^{-s}\, ds=\frac 1 {\Gamma(c+1)}t^ce^{-t}
+\frac 1 {\Gamma(c+1)}\int_0^ts^ce^{-s}\, ds.\end{equation}
To repeat the integration, we similarly observe that for all $\nu=0,1,2,\ldots$
\begin{equation}\label{j1}\frac1{\Gamma(c+\nu)}\int_0^t s^{c+\nu-1}e^{-s}\, ds=\frac 1 {\Gamma(c+1+\nu)}t^{c+\nu}e^{-t}+\frac
1 {\Gamma(c+1+\nu)}\int_0^ts^{c+\nu}e^{-s}\, ds,\end{equation}
and note that, for fixed $t>0$,
\begin{equation}\label{j2}\frac {\gamma(c+\nu,t)}{\Gamma(c+\nu+1)}\to 0,\qquad (\nu\to\infty).\end{equation}
In view of \eqref{j0}-\eqref{j2} we have
$$\frac {\gamma(c,t)}{\Gamma(c)}=t^c e^{-t}\sum_{\nu=0}^\infty \frac {t^{\nu}}{\Gamma(c+1+\nu)}
=t^ce^{-t}E_{1,c+1}(t),$$
and \eqref{cla} is proved.

It follows from \eqref{cla} that
\begin{equation}\label{cla2}I_{c,1}=\int_{0}^\infty\left[1-\frac {\gamma(c,t)}{\Gamma(c)}\right]\, dt=\frac 1 {\Gamma(c)}
\int_0^\infty \Gamma(c,t)\, dt.\end{equation}
To calculate this, we note that
\begin{equation}\label{cla3}\tfrac d {dt}(\Gamma(c+1,t)-t\Gamma(c,t))=-\Gamma(c,t),\end{equation}
which is immediate since $\tfrac d {dt}\Gamma(c,t)=-t^{c-1}e^{-t}$. Moreover, it is easy to see that
\begin{equation}\label{cla4}\lim_{t\to\infty}(\Gamma(c+1,t)-t\Gamma(c,t))=0.\end{equation}
Combining \eqref{cla2}-\eqref{cla4} we find that
$$I_{c,1}=\tfrac {\Gamma(c+1,0)}{\Gamma(c)}=c.$$
For the case when $-1<\re c<0$, we write
\begin{align*}
t^c e^{-t} E_{1,1+c}(t) &= t^c e^{-t}(\frac{1}{\Gamma(c+1)}+\sum_{j\geq 1}^{\infty}\frac{t^j}{\Gamma(j+c+1)})\\
&=\frac{t^ce^{-t}}{\Gamma(c+1)}+t^{c+1}e^{-t}E_{1,c+2}(t) = \frac{t^c e^{-t}}{\Gamma(c+1)}+\frac{\gamma(c+1,t)}{\Gamma(c+1)}.
\end{align*}
Thus, we have
$$I_{c,1} = \int_{0}^{\infty} (1 -
\frac{\gamma(c+1,t)}{\Gamma(c+1)}-\frac{t^ce^{-t}}{\Gamma(c+1)} )\,dt = (c+1)-1 =c,$$
finishing the proof of the lemma.
\end{proof}

\begin{rem} When $c$ is a positive integer, we can easily prove the identity $I_{c,1}=c$ as follows.
We have
$$\sum_{j=0}^\infty\frac {t^{j+c}}{\Gamma(j+1+c)}=e^t-P_c(t)$$
where
$$P_c(t)=\sum_{j=0}^{c-1}\frac {t^j}{j!}.$$
It follows that
$$I_{c,1}=\int_0^\infty P_c(t)e^{-t}\, dt=c,\qquad c=1,2,\ldots.$$
\end{rem}

In order to see the effect of the insertion near the boundary, we will apply the boundary fluctuation theorem in \cite{AM}.

For a suitable function $h$, let
$V_n = Q -2h/n$ be the perturbed potential and $\bfR_n^h$ be the corresponding $1$-point function. Write $\rho_n^h$ for the difference ${\bfR}_{n}^h - \bfR_n$.
We define for any suitable function $f$, a new function $f^S$ in the following way: $f^S=f$ on $S$ and $f^S$ equals the harmonic extension of $f|_{\d S}$ to $\hat{\C}\setminus S$ on that set.
\begin{mth}\label{crate}
Suppose that $f$ is a smooth test-function with compact support, which vanishes in a small, fixed neighborhood of $0$ and $h$ is $C^2$-smooth outside the origin and has compact support. Then there exists a small number $\alpha>0$ such that
\begin{equation}\label{rrate}\rho_n^h(f) = \frac{1}{2} \int \nabla f^S \bigcdot \nabla h^S + O(n^{-\alpha}),\quad n\to \infty.\end{equation}
\end{mth}

\begin{proof}
 The estimate \eqref{rrate} is given in \cite[Theorem 2.3]{AM}, except for the order of the error-term $O(n^{-\alpha})$.
To obtain this error-term, we just need to carefully examine the proof in \cite{AM}. In the following lines we will explain
how this can be seen, temporarily borrowing notation from \cite{AM}. (The reader who wishes to understand the following details
should thus have a copy of \cite{AM} at hand.)

Let $v$ be a suitable test-function, smooth in a neighbourhood of the boundary $\d S$, satisfying conditions (3.4-i)--(3.4-iii) from \cite{AM}.

The main tool behind the convergence of fluctuations in \cite{AM} is the limit form of Ward's identity:
$$2\int_{\C} D^h_n (v\d\dbar Q - \dbar v(\d \check{Q}-\d Q))=-\frac{1}{2}\sigma(\d v + 4v\d h)  +\eps^1_n[v]+\eps_{n}^2[v]+({\sigma}^h_n - \sigma)(\frac {\d v}2 + 2v\d h),$$
where $\eps_n^1[v]$ and $\eps_n^2[v]$ are certain error-terms given in \cite{AM}.

By Prop. 4.5 and Prop. 4.6 in \cite{AM}, we have
$$\eps^1_n[v] = O(n^{-1/2}[\log n]^4) \quad \text{and}\quad  \eps^2_n[v]=O(n^{-\beta/2}(\|v\|_{L^\infty}+ \|\nabla v\|_{L^\infty})),$$
for some small $\beta>0$, and it is easy to show that the term $({\sigma}^h_n - \sigma)(\d v/2 + 2v\d h)$, is $O(\delta_n)$ where $\delta_n=\log^2 n/\sqrt{n}$.
Hence it follows that
\begin{equation}\label{lw}
2\int_{\C} {D}^h_n (v\d\dbar Q - \dbar v(\d \check{Q}-\d Q))=-\frac{1}{2}\sigma(\d v) - 2\sigma(v\d h ) + O(n^{-\alpha})\end{equation}
for some $\alpha>0$.

Now we assume that $f$ is a suitable function supported near the boundary. Write $f=f_+ +f_0$ and $v=v_+ + v_0$ where
$$v_+ = \frac{\dbar f_+}{\Lap Q}1_S,\quad \text{and} \quad v_0 = \frac{\dbar f_0}{\Lap Q}1_S + \frac{f_0}{\d (Q-\check{Q})}1_{\C\setminus S}$$
 as in \cite[Section 5]{AM}. Then \eqref{lw} implies that,
$$2\int_\C {D}^h_n \cdot \dbar f = -\frac{1}{2}\sigma(\d v) - 2\sigma(v\d h)+ O(n^{-\alpha}).$$
Hence the fluctuation term ${\nu}^h_n(f)$ obeys the asymptotic
$$\nu^h_{n}(f) = \frac{1}{4}\sigma(\d v)+\sigma(v\d h) + O(n^{-\alpha}),$$
and we obtain
\begin{align*}
\rho_n^h(f)&={\nu}^h_n(f)-\nu_n(f) = \sigma(v\d h)+O(n^{-\alpha})\\
&= 2\int_S \dbar f \cdot \d h + O(n^{-\alpha})
= \frac{1}{2} \int \nabla f^s \bigcdot \nabla h^s + O(n^{-\alpha}).
\end{align*}
The general case $f=f_0+f_++f_-$ is immediate from this, as in \cite[Section 5]{AM}.
\end{proof}

To prove Theorem \ref{thm1}, we consider the cases when $h=G$ and when $h=\ell$. The function $G$ is an extension of the Green's function, which is smooth on $\C^*$ and has compact support. For the case of "pure log-normalization", the limit form of Ward's identity in \cite{AM} can be applied and Theorem \ref{crate} holds for $\ell$ since the function $\ell$ does not grow too fast and $\d \ell(z) \sim z^{-1}$ near $\infty$.

\begin{proof}[Proof of Theorem \ref{thm1}] Assume that the external potential $V_n$ has the Green's form in \eqref{pot2}. We then have $\rho_n = \rho_n ^{G}$.

\begin{lem} \label{pelm}
Suppose that $f$ is a smooth test-function which vanishes in some small, fixed neighbourhood of $0$, and let
$\omega_0$ be harmonic measure for $S$ evaluated at $0$. Then
\begin{equation*}
\lim_{n\to\infty} \rho_n(f)\to c\omega_0(f).\end{equation*}\end{lem}

\begin{proof} The assumptions on $f$ imply that we can apply the boundary fluctuation theorem \ref{crate}. Also let $G$ be the Green's function in \eqref{gu}, so $G^S=0$ in $\C\setminus S$.

The asymptotic
formula for the variance of fluctuations in Theorem \ref{crate} implies that, for
the test-functions $f$ under consideration,
\begin{align*}
\lim_{n\to\infty}\rho_n(f)&=\frac c {2}\int_S \nabla f\bigcdot(-\nabla G)\, dA
\\
&=-\frac c {2\pi}\int_{\d S}f\frac {\d G}{\d n}\, ds +2c\int f\Lap G\, dA\\
&=-\frac c {2\pi}\int_{\d S}f\frac {\d G}{\d n}\, ds,
\end{align*}
since $f$ vanishes near $0$. Here $\d/\d n$ is differentiation in the outwards normal direction.
We refer to the argument in \cite[p. 76]{AHM2} for details about the calculation.

The lemma follows, since $-\frac 1 {2\pi}\frac {\d G}{\d n}\, ds$ is the harmonic measure $\omega_0$
(see \cite{GM}).
\end{proof}

Combining lemmas \ref{pelm} and \ref{vott}, we conclude the proof of
part \ref{fi} of Theorem \ref{thm1}.

Now assume that the external potential $V_n$ has the pure $\log$-form, i.e.,
$$V_n=Q-(2c/n)\ell,\qquad \ell(\zeta)=\log|\zeta|.$$
In this case, we have $\rho_n = \rho_n^\ell$. Let $G_\infty$ be the Green's function of
$\C\setminus S$ with pole at $\infty$ (so $G_\infty(\zeta)\sim\log|\zeta|$ as $\zeta\to\infty$). We consider
$G_\infty$ as being extended to $\C$ in some smooth way. Note that
$$\ell^S(\zeta)=\ell(\zeta)-G_\infty(\zeta),\qquad \zeta\in \C\setminus S,$$
and that $\ell^S$ is harmonic on $\hat{\C}\setminus S$.

We now fix a function $f\in C_0^\infty$ which vanishes near $0$ and apply the result of \cite{AM} as in the
proof of Lemma \ref{pelm}. The result is this time that
\begin{align*}\lim_{n\to\infty}(\tilde{\bfR}_n-\bfR_n)(f)&=\frac c {2}\int_S \nabla f\bigcdot\nabla \ell
+\frac c {2}\int_{S^c} \nabla f^S\bigcdot\nabla \ell^S\\
&=\frac c {2\pi}\int_{\d S}f\frac {\d G_\infty}{\d n}\, ds=c\int_{\d S}f\, d\omega_\infty,
\end{align*}
where $\omega_\infty$ is harmonic measure of $\C\setminus S$ evaluated at $\infty$. Combining with Lemma
\ref{vott}, we conclude the proof of part \ref{se} of Theorem \ref{thm1}.
\end{proof}

\section{A Central Limit Theorem}\label{CLT}
In this section we prove Theorem \ref{mt7}.
Our strategy is to first give the proof in the simplified model Mittag-Leffler case, and after that we will
adapt the proof to a general algebraic insertion potential.

\begin{lem} \label{voxda} Suppose that $V_n^{\sharp}$ is the model Mittag-Leffler potential
\begin{equation}\label{short}V_n^{\sharp}=|\zeta|^{2\lambda}-\frac {2c} n\log|\zeta|.\end{equation}
Then the random variables
$$X_n=\tfrac 2 {\sqrt{\log n}}(\tr_n\ell-\bfE_n \tr_n\ell)$$
converge in distribution to the centered normal distribution with
variance $1/\lambda$.
\end{lem}

We remind the reader of the notation $\ell(\zeta)=\log|\zeta|$,
$\tr_n\ell=\sum\ell(\zeta_j)$ where $\{\zeta_j\}_1^n$ is a random sample. $\bfE_n$ denotes the expectation with respect to the potential $V_n^{\sharp}$, see \eqref{bogi}.
We also recall that the rescaled one-point function of \eqref{short} equals to
\begin{equation} \label{nfor} R_n^{\sharp}(z)=\lambda\sum_{j=0}^{n-1}\frac {|z|^{2(j+c)}}
{\Gamma(\frac {j+c+1} \lambda)}\, e^{-|z|^{2\lambda}}.
\end{equation}
(See the example in Section \ref{imr}.)

\begin{proof}[Proof of the lemma.]
We shall apply the variational approach from \cite{AM} with the
scale of potentials
$$V_{n,t}^{\sharp}=V_n^{\sharp}-\frac {2c_{n,t}}n\log|\zeta|,\quad c_{n,t}=\frac t{\sqrt{\log n}}.$$
Here $t$ is a fixed real constant and $n$ is large enough so that $c_{n,t}>-1$.

The method uses the function
$$F_n(t)=\log\bfE_n\exp(tX_n),$$
where $\bfE_n$ is expectation with respect to the potential $V_n^{\sharp}$.

As in \cite{AHM2,AM} we note that
$$F_n'(t)=\frac {\bfE_n[X_ne^{tX_n}]}{\bfE_n[e^{tX_n}]}=
\frac {\int X_n e^{2c_{n,t}\sum\ell(\zeta_j)}\, d\bfP_n}
{\int e^{2c_{n,t}\sum\ell(\zeta_j)}\, d\bfP_n}=
\bfE_{n,t}[X_n],$$
where $\bfE_{n,t}$ is expectation with respect to potential $V_{n,t}^{\sharp}$.
Hence
\begin{equation}\label{xl}F_n'(t)=\frac 1 {\sqrt{\log n}}\int_\C\ell\cdot (\bfR_{n,t}^{\sharp}-\bfR_n^{\sharp})\, dA\end{equation}
where $\bfR_{n,t}^{\sharp}$ and $\bfR_{n}^{\sharp}$ are $1$-point functions with respect to $V_{n,t}^{\sharp}$ and $V_{n}^{\sharp}$, respectively.

Inserting explicit expressions (see \eqref{nfor}) we obtain
\begin{align*}
F_n'(t)
&=\frac {2\lambda} {\sqrt{\log n}}\int_\C \ell(zn^{-1/2\lambda})\sum_{j=0}^{n-1}(\frac
{(|z|^2)^{j+c+c_{n,t}}}{\Gamma(\frac{j+1+c+c_{n,t}}{\lambda})}-
\frac {(|z|^2)^{j+c}}{\Gamma(\frac{j+c+1}{\lambda})})e^{-|z|^{2\lambda}}\, dA(z)\\
&=\frac {\lambda^{-1}} {\sqrt{\log n}}\int_0^\infty(\log s-\log n)\sum_{j=0}^{n-1}(
\frac {s^{\frac{j+1+c+c_{n,t}}{\lambda}-1}}{\Gamma(\frac{j+1+c+c_{n,t}}{\lambda})}
-\frac {s^{\frac{j+1+c}{\lambda}-1}}{\Gamma(\frac{j+1+c}{\lambda})})e^{-s}\, ds.
\end{align*}

Let $\psi$ be the polygamma function,
$$\psi(x)=\frac {\Gamma'(x)}{\Gamma(x)}=\frac 1 {\Gamma(x)}\int_0^\infty s^{x-1}
e^{-s}\log s\, ds.$$
The computations above show that
\begin{equation*}F_n'(t)=\frac {\lambda^{-1}} {\sqrt{\log n}}\sum_{j=1}^{n}(\psi(
\frac{j+c+c_{n,t}}{\lambda})-\psi(\frac{j+c}{\lambda})).\end{equation*}
We now use Taylor's formula to write
\begin{equation*}\psi(\frac{j+c+c_{n,t}}{\lambda})-\psi(\frac{j+c}{\lambda})=
\frac{c_{n,t}}{\lambda}\psi'(\frac{j+c}{\lambda})+\frac{1}{2}
(\frac {c_{n,t}}{\lambda})^2\psi''(\xi_j),
\end{equation*}
where $\xi_j$ is some number between $(j+c)/\lambda$ and $(j+c+c_{n,t})/\lambda$, and have

\begin{align}\label{by}
F_n'(t) = \frac{t}{\lambda^2\log n}\sum_{j=1}^{n}\psi'(\frac{j+c}{\lambda}) +\frac{t^2}{2(\lambda^2\log n)^{3/2}}\sum_{j=1}^{n}\psi''(\xi_j)+o(1).\end{align}
To proceed with this,
we note the following lemma.

\begin{lem} \label{prelem} $\psi$ and $\psi'$ have the series expansions
$$\psi(x+1)=-\gamma+\sum_{k=1}^\infty(\frac 1 k-\frac 1 {k+x})
\quad \text{and}\quad \psi'(x+1)=\sum_{k=1}^\infty\frac 1 {(k+x)^2},$$
for $x\not\in\{-1,-2,\ldots\}$, where $\gamma$ is the Euler constant.
\end{lem}

The proof of the lemma is an immediate consequence of Weierstrass' form of the Gamma function,
$$\Gamma(x+1)=e^{-\gamma x}\prod_{k=1}^\infty(1+\frac x k)^{-1}e^{x/k}.$$

\medskip

Since $\psi'$ is decreasing for $x>0$, we obtain
\begin{equation}\label{intcom}
\int_{1}^{n+1} \psi'(\frac{t+c}{\lambda})\,dt \leq \sum_{j=1}^{n}\psi'(\frac{j+c}{\lambda})\leq \psi'(1) + \int_{1}^{n} \psi'(\frac{t+c}{\lambda})\, dt.
\end{equation}
Lemma \ref{prelem} shows that
\begin{equation}\label{estpg}
\psi(x+1)=\log x+\frac 1 {2x}+O(x^{-2})\quad
\text{and}\quad \psi'(x+1)=\frac 1 x+O(x^{-2})
\end{equation}
as $x\to\infty$. Then it follows from \eqref{intcom} and \eqref{estpg} that
$$\frac{t}{\lambda^2\log n}\sum_{j=1}^{n}\psi'(\frac{j+c}{\lambda}) = \frac{t}{\lambda} +o(1) \quad \text{as} \quad n\to \infty.$$
We also have
$$\frac {t^2}{2(\lambda^2\log n)^{3/2}}\sum_{j=1}^{n}\psi''(\xi_j)\to 0\quad \text{as}\quad
n\to\infty$$
since $\psi''(x)= -1/x^2 + O(x^{-3})$ for large $x$.
By \eqref{by} we now obtain that $F_n'(t)\to t/\lambda$ as $n\to\infty$, and it is easy to
see that our estimates give locally uniform convergence on $\R$. Since $F_n(0)=0$, we conclude
that $F_n(t)=t^2/2\lambda$ as $n\to\infty$, i.e.,
$$\bfE_n e^{tX_n}\to e^{t^2/2\lambda}.$$
It is well-known that this implies convergence in distribution to the normal distribution with mean zero and variance $1/\lambda$.
\end{proof}

We now generalize to an arbitrary potential of the form
\begin{equation}\label{macl}V_n(\zeta)=|\zeta|^{2\lambda}+h(\zeta)-\frac {2c} n\log|\zeta|\end{equation}
where $h$ is a harmonic polynomial (in some neighbourhood of the droplet). Let $\bfR_n$ and $R_n$ denote the $1$-point function and the rescaled one with respect to $V_n$. We know that the rescaled $1$-point function $R_n$ converges to
$$R(z)=\lambda\cdot E_{\frac 1 \lambda,\frac {1+c}\lambda}(|z|^2)e^{-V_0(z)},\qquad (V_0=|z|^{2\lambda}-2c\log|z|).$$
However, in order to apply the argument in Lemma \ref{voxda}, it will be better to work with the $1$-point function $$R_n^{\sharp}(z)=\lambda\sum_{j=0}^{n-1}\frac {|z|^{2(j+c)}}{\Gamma(\frac {j+c+1}\lambda)}e^{-|z|^{2\lambda}}$$
with respect to $V_n^{\sharp}(\zeta) = |\zeta|^{2\lambda}-\frac{2c}{n}\log|\zeta|$, which is the model Mittag-Leffler case considered in Lemma \ref{voxda}.
Note also that Theorem \ref{mt22} implies the convergence $R_n-R_n^{\sharp}\to 0$ in the sense of distributions on $\C$ and locally uniformly on $\C^*$.

Now fix a suitable $t\in\R$ and put
$$c_{n,t}=\frac t{\sqrt{\log n}}.$$
We consider the perturbed potential
$$V_{n,t}=V_n-\frac {2c_{n,t}}n\log|\zeta|$$
and we write $\bfE_{n,t}$, $\bfR_{n,t}$ for the corresponding expectation and $1$-point function, respectively. We also write $R_{n,t}$, $R_{n,t}^{\sharp}$ for the corresponding rescaled 1-point functions, i.e., $R_{n,t}$ and $R_{n,t}^{\sharp}$ are the rescaled 1-point functions associated with the potentials $V_{n,t}$ and $V_{n,t}^{\sharp}$, respectively.

As before, we introduce
$$F_n(t)=\log\bfE_n\exp(tX_n)$$
where
$$X_n=\tfrac 2 {\sqrt{\log n}}(\tr_n\ell-\bfE_n\tr_n \ell),\qquad (\ell(\zeta)=\log|\zeta|).$$
We know that
$F_n'(t)=\bfE_{n,t}(X_n),$
so
\begin{equation}\label{ex1}F_n'(t)=\frac 1 {\sqrt{\log n}}\int_\C
2\ell(\zeta)\cdot (\bfR_{n,t}(\zeta)-\bfR_n(\zeta)).\end{equation}

Consider the set $\calR_n=\{\zeta\in\Int S;\, |\zeta|>r_n(\log n)^{\frac{1}{\lambda}},\, \dist(\zeta,\d S)>r_n(\log n)^{\frac{1}{\lambda}}\}$.
We shall use the uniform estimate
\begin{equation}\label{oj}\sup_{\zeta\in \calR_n}|\bfR_n(\zeta)-n\Lap Q_0(\zeta)|\le Ce^{-\alpha\log^{2}n}\end{equation}
with some constants $C,\alpha>0$. The detailed proof
of \eqref{oj} depends on an adaptation of the technique of approximate Bergman projections, which is postponed to the next section (see Theorem \ref{asymthm} below).

By the same token, the obvious counterpart to \eqref{oj} is true also for the difference $\bfR_{n,t}-n\Lap Q_0$, so we obtain the result that
\begin{equation}\label{boxer}\sup_{\zeta\in\calR_n}|\bfR_{n,t}(\zeta)-\bfR_n(\zeta)|\le
Ce^{-\alpha\log^{2}n}.\end{equation}

In the vicinity of the boundary $\d S$, we do not have such a strong uniform control, but due to our discussion of balayages in Section \ref{sec5}
we know that $\bfR_{n,t}-\bfR_n\sim c_{n,t}\omega_\infty$ there,
where $\omega_\infty$ is harmonic measure of $\d S$ evaluated at $\infty$.

Outside a neighborhood of $S$, the one point function is controlled by the exterior estimate, see e.g., \cite[Section 4.1.1]{AM}. For any $m>0$, there exists a constant $C_m$ such that
$\bfR_{n,t}(\zeta) \leq C_m n^{- m} (1 + |\zeta|)^{-3}$ for all $\zeta$ outside a fixed neighborhood of $S$.

Combining these asymptotic estimates, we find that
\begin{equation}\label{goodest}\int_{|\zeta|>r_n(\log n)^{\frac{1}{\lambda}}}|\ell(\zeta)|\cdot |\bfR_{n,t}(\zeta)-\bfR_n(\zeta)|=O(\tfrac 1{\sqrt{\log n}}).\end{equation}
Here the $O$-constant is proportional to $|t|$. The convergence $\bfR_{n,t}-\bfR_n \to c_{n,t}(\omega_\infty-\delta_0)$ shown in Theorem \ref{thm1} gives
\begin{equation}\label{finest} \frac{2\log r_n}{\sqrt{\log n}}
\int_{|\zeta|>r_n(\log n)^{\frac{1}{\lambda}}} (\bfR_{n,t}-\bfR_n) = c_{n,t}\frac{2\log r_n}{\sqrt{\log n}}+o(1) = -\frac{t}{\lambda} +o(1).
\end{equation}
Here, the error bound $o(1)$ is obtained by the error bound of the convergence in Theorem \ref{crate}.

Let us denote
$$J_{n,t}=\frac 2 {\sqrt{\log n}}\int_{|\zeta|\le r_n(\log n)^{\frac{1}{\lambda}}}\ell(\zeta)\cdot(\bfR_{n,t}(\zeta)-\bfR_n(\zeta)).$$
By \eqref{goodest} we have that
$$|F_n'(t)-J_{n,t}|\le\frac C{\log n}.$$
Let us introduce the function
$$
\ell_n(z):=\tfrac {2\log|z|}{\sqrt{\log n}}\, \chi_{D(0;(\log n)^{{1}/{\lambda}})}(z).$$
Changing variables in \eqref{ex1} by $\zeta=r_nz$ and the asymptotic expansion \eqref{finest} give
\begin{align*}
F_n'(t)&=\frac 2 {\sqrt{\log n}}\int_{|z|\le(\log n)^{\frac{1}{\lambda}}} \log|r_n z|\cdot(R_{n,t}(z)-R_n(z))+O(1/\log n)\\
&=(R_{n,t}-R_n)(\ell_n) + \frac{t}{\lambda} +o(1).
\end{align*}
We split the last expression as
\begin{align*}(R_{n,t}-R_n)(\ell_n)&=(R_{n,t}-R_{n,t}^\sharp)(\ell_n)-
(R_n-R_n^\sharp)(\ell_n)+
(R_{n,t}^\sharp-R_n^\sharp)(\ell_n)\\
&=A_{n,t}-B_{n}+C_{n,t}.
\end{align*}
We prove by Lemma \ref{voxda} that
$$C_{n,t}= o(1) \quad \text{as}\quad n\to\infty.$$
Indeed, by the approximations above applied to suitable model Mittag-Leffler ensembles, we obtain
$$\frac{2}{\sqrt{\log n}}\int_{\C} \log|r_nz|\cdot(R^{\sharp}_{n,t}(z)-R_n^{\sharp}(z)) = C_{n,t}+\frac{t}{\lambda}+o(1)$$ and
the integral converges to $t/\lambda$ by Lemma \ref{voxda}.

We now want to show that $A_{n,t},B_n\to 0$ as $n\to\infty$. For this, we recall the following estimates, which are proved in \cite{AS2}.

 First of all, note that $R_n - R_n^\sharp \to 0$ and $R_{n,t}-R_{n,t}^\sharp\to 0$ locally uniformly on $\C^*$ and in the sense of distribution on $\C$. (See Theorem \ref{mt2} and Theorem \ref{mt22}.) Also note that Theorem \ref{mt1} implies that these $1$-point functions are dominated near the origin ($R_n(z)\leq C|z|^{2c}$, $|z|\leq 1$). Thus, there is no problem to estimate the integral of
 $\ell_n\cdot (R_n-R_n^\sharp)$ over a large compact disk $|z|\le M$.

 Suppose now that $M\le|z|\le (\log n)^{\frac{1}{\lambda}}$. There are then, again by \cite[Theorem 4]{AS2}, constants $C,\alpha>0$ such that (i)
$|R_n(z)-\Lap Q_0(z)|\le C|z|^{2\lambda-2}e^{-\alpha|z|^{2\lambda}}$, (ii)
 $|R_n^\sharp(z)-\Lap Q_0(z)|\le C|z|^{2\lambda-2}e^{-\alpha|z|^{2\lambda}}$, (iii)
 $|R(z)-\Lap Q_0(z)|\le C|z|^{2\lambda-2}e^{-\alpha|z|^{2\lambda}}$.

All in all, using that $|\ell_n(z)|=O(\log\log n/\sqrt{\log n})$ when $M\le |z|\le (\log n)^{\frac{1}{\lambda}}$, we obtain the estimate
$$|B_n|=o(1)+o(1)\int_{M\le |z|\le (\log n)^{\frac{1}{\lambda}}}|z|^{2\lambda}e^{-\alpha|z|^{2\lambda}},$$
i.e. $B_n\to 0$ as $n\to\infty$. Similarly, $A_{n,t}\to0$ as $n\to\infty$.

\section{Asymptotics for the 1-point function} \label{asyl}

In this section, we will be dealing with the class of \textit{generalized Hele-Shaw potentials}, by which we mean potentials of the form
\begin{equation}\label{genhs}V_n(\zeta)=|\zeta|^{2\lambda}-\tfrac {2c}n\log|\zeta|+2\re H(\zeta)\end{equation}
 where $c>-1$ and where $H$ is holomorphic in a neighbourhood of the droplet. We assume, as always, that $0$ is an interior point of the droplet.

Below we will denote by $r_n=n^{-1/2\lambda}$.
We define for each $n$ a set $\calR_n\subset \Int S$ of ``regular bulk points'' by
$$
\calR_n=\{\zeta\in\Int S;\, |\zeta|>r_n(\log n)^{\frac{1}{\lambda}},\quad \dist(\zeta,\d S)>r_n(\log n)^{\frac{1}{\lambda}}\}.$$

We have the following theorem, which generalizes a result from the ``ordinary'' Hele-Shaw case $Q=|\zeta|^2+2\re H(\zeta)$ (see \cite[Theorem 2.2]{A2}).

\begin{mth} \label{asymthm} If $\zeta\in \calR_n$ then there are numbers $C,\alpha>0$ such that
$$|\bfR_n(\zeta)-n\lambda^2|\zeta|^{2\lambda-2}|\le Ce^{-\alpha\log^{2}n}.$$
\end{mth}

Note that Theorem \ref{asymthm} in particular completes our  argument for the CLT in Section \ref{CLT} (i.e., we obtain the missing estimate \eqref{oj}).

Before discussing the proof, it is instructive to compare the result with other kinds of approximations.

Using a recursive scheme, reproduced in \cite{A2} in the present setting,
it is not hard to see that for $\zeta\in\Int S\setminus\{0\}$, the coefficients $b_j$ in the formal (Tian-Catlin-Zelditch type) expansion \begin{equation}\label{TCZ}
\bfR_n(\zeta)=nb_0(\zeta)+b_1(\zeta)+n^{-1}b_2(\zeta)+\cdots\end{equation}
are just $b_0(\zeta)=\Lap |\zeta|^{2\lambda}=\lambda^2|\zeta|^{2\lambda-2}$, $b_1=\frac 1 2 \Lap\log b_0=0$, and and $b_j=0$ for $j\ge 2$.
 It could thus be surmised that the approximation $\bfR_n\sim nb_0$ should hold to a very good accuracy in the domain $\calR_n$.
We shall prove that this is indeed the case, by adapting a method from \cite[Section 6]{A2}.

\smallskip
To prepare the ground, we fix a sequence $\zeta_n$ with $\zeta_n\in \calR_n$, and we put $$
\delta_n:=c_\lambda n^{-1/2}|\zeta_n|^{1-\lambda}\log n
$$
for a small enough constant $c_\lambda$ less than $1/\max\{\lambda,10\}$. If $\lambda =1$, the distance $\delta_n = O(n^{-1/2}\log n)$ is chosen so as to be ``sufficiently large'' compared with
the typical inter-particle distance $n^{-1/2}$.
In general, our choice of $\delta_n$ is used to deduce the inequality \eqref{gelog}, which will be seen to imply the exponential decay in Theorem \ref{asymthm}.

We can easily see that
$${\delta_n}/{|\zeta_n|}=c_\lambda n^{-1/2}|\zeta_n|^{-\lambda}\log n \leq c_\lambda,\quad \zeta_n \in \calR_n.$$
We also fix a sequence of cut-off functions $\chi_n$ with $\chi_n=1$ on $D(\zeta_n;2\delta_n)$, $\chi_n=0$ outside $D(\zeta_n;3\delta_n)$ and
$\|\dbar\chi_n\|_{L^2}\le C$ (independent of $n$).

Finally, when $\phi:\C\to \R_+$ is a suitable weight function, we denote the scalar product in the space $L^2(e^{-\phi})=L^2(e^{-\phi},dA)$ by
$$\langle f,g\rangle_\phi=\int_\C f\bar{g}e^{-\phi}\, dA.$$

Polarizing in the formula for $V_n$, we define for $\zeta,\eta\in D(\zeta_n,3\delta_n)$
$$V_n(\zeta,\eta)=\zeta^\lambda\bar{\eta}^\lambda+H(\zeta)+\bar{H}(\eta)-\tfrac c n(\log\zeta+\overline{\log}\,\eta),$$
where $\log$ is some determination of the logarithm in the disk $D(\zeta_n;3\delta_n)$. The main fact to remember below is that
$$2\re V_n(\zeta,\eta)-V_n(\zeta)-V_n(\eta)=-|\zeta^\lambda-\eta^\lambda|^2.$$

Following the idea in \cite[Section 6]{A2} we note that we can write
$$V_n(\zeta,\eta)-V_n(\eta,\eta)=(\zeta^\lambda-\eta^\lambda)\bar{\eta}^\lambda+\text{holomorphic in $\eta$}.$$

Next note that for each fixed $\zeta\ne 0$, we have in the sense of distributions on $D(\zeta_n,3\delta_n)$
$$\dbar_\eta(\tfrac 1 {\zeta^\lambda-\eta^\lambda})=
-\tfrac1 {\lambda\zeta^{\lambda-1}}\,\delta_\zeta(\eta).$$
(For if $\zeta,\eta\in D(\zeta_n,3\delta_n)$, then $\zeta^\lambda-\eta^\lambda=0\iff\zeta=\eta$.)

Now assume $\zeta\in D(\zeta_n;\delta_n)$. By Cauchy's formula, and the above, we have for each function
$u$ holomorphic and bounded in
$D(\zeta_n;3\delta_n)$,
\begin{align*}u(\zeta)&=\lambda\zeta^{\lambda-1}\int\frac {\dbar(u(\eta)\chi_n(\eta)e^{n(V_n(\zeta,\eta)-V_n(\eta,\eta))})}
{\zeta^\lambda-\eta^\lambda}\, dA(\eta)\\
&=n\int u(\eta)\chi_n(\eta)\lambda^2(\zeta\bar{\eta})^{\lambda-1}
 e^{n(V_n(\zeta,\eta)-V_n(\eta,\eta))}\, dA(\eta)\\
 &+\lambda\zeta^{\lambda-1}\int\frac {u(\eta)\dbar\chi_n(\eta)}{\zeta^\lambda-\eta^\lambda}
e^{n(V_n(\zeta,\eta)-V_n(\eta,\eta))}\, dA(\eta)\\
&=:\mathrm{I}_n^*u(\zeta)+\mathrm{II}_n^*u(\zeta).
\end{align*}

We now come to an important observation. Since $|\zeta-\eta|\ge \delta_n$ when $|\zeta-\zeta_n|\le\delta_n$ and $\dbar\chi_n(\eta)\ne 0$, we have by choosing $c_\lambda$ small enough and by Taylor's formula
\begin{equation}\label{gelog}
n|\zeta^\lambda-\eta^\lambda|^2\ge \const \,n |\zeta_n|^{2\lambda-2} \delta_n^{2}\ge 2\alpha\log^{2}n
\end{equation}
for some constant $\alpha>0$.

Using this, we now estimate the term $\mathrm{II}_n^*u(\zeta)$ as follows
\begin{align*}|\mathrm{II}^*_nu(\zeta)|&
\le C\int|\tfrac {u(\eta)\dbar\chi_n(\eta)}{\zeta-\eta}| e^{-n|\zeta^\lambda-\eta^\lambda|^2/2-nV_n(\eta)/2+nV_n(\zeta)/2}\, dA(\eta)\\
&\le Ce^{nV_n(\zeta)/2}\delta_n^{-1}e^{-\alpha\log^{2}n}\int_{|\eta-\zeta|\ge \delta_n}
|u(\eta)|e^{-nV_n(\eta)/2}|\dbar\chi_n(\eta)|\, dA(\eta)\\
&\le \delta_n^{-1}e^{-\alpha \log^{2}n}\|\dbar\chi_n\|_{L^2}e^{nV_n(\zeta)/2}\|u\|_{nV_n},
\end{align*}
where we used the Cauchy-Schwarz inequality.
We summarize our findings in the following lemma.

\begin{lem} \label{dif} There are constants $C,\alpha>0$ such that if $\zeta\in D(\zeta_n,\delta_n)$, then
$$|u(\zeta)-\mathrm{I}_n^*u(\zeta)|\le C\delta_n^{-1}e^{-\alpha\log^{2}n}\|u\|_{nV_n}e^{nV_n(\zeta)/2}.$$
\end{lem}

We now define (for suitable points $\zeta,\eta$ near $\zeta_n$) the approximate kernel $\bfL_n^*$ by
$$\bfL_{n,\zeta}^*(\eta)=\bfL_{n}^*(\eta,\zeta)=n\chi_n(\eta)\lambda^2(\eta\bar{\zeta})^{\lambda-1}
e^{nV_n(\eta,\zeta)}.$$
Then the operator $\mathrm{I}_n^*u(\zeta)$ defined above is just
$$\mathrm{I}_n^*u(\zeta)=\langle u,\bfL_{n,\zeta}^*\rangle_{nV_n}.$$

Let $\bfL_n(z,w)$ be the reproducing kernel for the space $\Pol(n)$ with norm of $L^2(e^{-nV_n})$. By Lemma \ref{dif}, we have the estimate
$$|\bfL_{n,\eta}(\zeta)-\mathrm{I}_n^*\bfL_{n,\eta}(\zeta)|\le
C\delta_n^{-1}e^{-\alpha\log^{2}n}\|\bfL_{n,\eta}\|_{nV_n}e^{nV_n(\zeta)/2}.$$

The following simple lemma is an adaptation of e.g.~\cite[Lemma 3.1]{AKM}.

\begin{lem}\label{lemsubm}
 Let $\eta\in D(\zeta_n,\delta_n)$.
Suppose that $u$ is analytic in $D:=D(\eta,tn^{-{1}/{2}}|\eta|^{1-\lambda})$
and let $f=ue^{-nV_n/2}$. Then there is a constant $C$ depending only on $t$ such that
$|f(\eta)|^2\le Cn|\eta|^{2\lambda-2}\int_D|f|^2.$
\end{lem}

\begin{proof} Let $a>0$ be a constant and form the function $$F_n(z)=f(\eta+z n^{-{1}/{2}}|\eta|^{1-\lambda})\cdot e^{a|z|^2/2}.$$ Then
$\Lap \log|F_n(z)|^2\ge -\lambda^2 |\eta|^{2-2\lambda}\cdot |\eta+zn^{-{1}/{2}}|\eta|^{1-\lambda}|^{2\lambda-2}+a>0$ for $|z|\le t$ if $a$ is large enough. This implies that $|F_n|^2$ is (logarithmically) subharmonic in $D(0,t)$. The following submean-value property
\begin{align*}
|F_n (0)|^2 \leq C' \int_{D(0,t)} |F_n(z)|^2 dA(z)
\end{align*}
implies the inequality
\begin{align*}
|f(\eta)|^2 &\leq C' \int_{D(0,t)} |f(\eta + zn^{-1/2}|\eta|^{1-\lambda})|^2 e^{a|z|^2}dA(z) \\
&\leq C \,n |\eta|^{2\lambda-2} \int_{D} |f(\zeta)|^2 dA(\zeta),
\end{align*}
where $C'$ and $C$ are some constants depending on $t$.
\end{proof}

If $\eta\in D(\zeta_n,\delta_n)$, then recalling that
$$\bfL_n(\eta,\eta)=\sup\{|u(\eta)|^2;\, u\in \Pol(n),\, \|u\|_{nV_n}\le 1\},$$ we get by Lemma \ref{lemsubm} that
\begin{align*}\|\bfL_{n,\eta}\|_{nV_n}^2=\bfL_n(\eta,\eta) \le Cn|\eta|^{2\lambda-2}e^{nV_n(\eta)}.\end{align*}
More precisely, Lemma \ref{lemsubm} implies that for all $u\in \Pol(n)$ with $\|u\|_{nV_n}\leq 1$, $$|u(\eta)|^2 \leq Cn|\eta|^{2\lambda-2}e^{nV_n(\eta)}\int_{D}|u|^2 e^{-nV_n} dA \leq Cn|\eta|^{2\lambda-2}e^{nV_n(\eta)},$$
which gives the bound of $\bfL_n(\eta,\eta)$. We conclude that
\begin{equation}\label{e11}|\bfL_n(\zeta,\eta)-\mathrm{I}_n^*\bfL_{n,\eta}(\zeta)|\le C\sqrt{n}\delta_n^{-1}|\eta|^{\lambda-1}e^{-\alpha\log^{2} n} e^{n(V_n(\zeta)+V_n(\eta))/2}.\end{equation}

We now note that
$$\overline{\mathrm{I}_n^*\bfL_{n,\zeta}(\eta)}=\mathrm{P}_n\bfL_{n,\eta}^*(\zeta),$$
where $\mathrm{P}_n$ is the polynomial Bergman projection,
$$\mathrm{P}_nf(\zeta)=\langle f,\bfL_{n,\zeta}\rangle_{nV_n}=
\int f(\eta)\bfL_n(\zeta,\eta)e^{-nV_n(\eta)}\, dA(\eta).$$

Applying \eqref{e11} now gives.

\begin{lem}
For $\zeta,\eta\in D(\zeta_n,\delta_n)$ we have
$$|\bfL_n(\zeta,\eta)-\mathrm{P}_n\bfL_{n,\eta}^*(\zeta)|\le Cn|\zeta_n|^{2\lambda-2}e^{-\alpha\log^{2} n}e^{n(V_n(\zeta)+V_n(\eta))/2}.$$
\end{lem}

Now fix $\eta\in D(\zeta_n,\delta_n)$ and introduce the function
$$u_{n,\eta}(\zeta)=\bfL_{n,\eta}^*(\zeta)-\mathrm{P}_n\bfL_{n,\eta}^*(\zeta).$$
Observe that $u_{n,\eta}$ is the $L^2(e^{-nV_n})$-minimal
solution to the problem $\dbar u=\dbar \bfL_{n,\eta}^*$ and $u-\bfL_{n,\eta}^*\in \Pol(n)$.
In view of a standard H\"{o}rmander estimate in \cite[Lemma 5.2]{A2}, we infer that there is a constant $C$ such that
\begin{equation}\label{w-est}
\|u_{n,\eta}\|_{nV_n}\le Cn^{-{1}/{2}}|\eta|^{1-\lambda}\|\dbar\bfL_{n,\eta}^*\|_{nV_n}.\end{equation}

\begin{lem}
 Let $\zeta,\eta\in D(\zeta_n,\delta_n)$. Then
$$|u_{n,\eta}(\zeta)|\le Cn|\zeta_n|^{2\lambda-2}e^{-\alpha\log^{2}n}e^{n(V_n(\zeta)+V_n(\eta))/2}.$$
\end{lem}

\begin{proof} First we fix $\eta\in D(\zeta_n,\delta_n)$ but let $\zeta$ be
unrestricted.
We have that
$$\dbar u_{n,\eta}(\zeta)=\dbar\bfL_{n,\eta}^*(\zeta)=\dbar\chi_n(\zeta) n\lambda^2
(\zeta\bar{\eta})^{\lambda-1}e^{nV_n(\zeta,\eta)}.$$
This implies that
$$|\dbar u_{n,\eta}(\zeta)|^2e^{-nV_n(\zeta)}=|\dbar\chi_n(\zeta)|^2\lambda^4n^2|\zeta\eta|^{2\lambda-2}e^{-n|\zeta^\lambda-\eta^\lambda|^2}
e^{nV_n(\eta)}.$$
Since $n|\zeta^\lambda-\eta^\lambda|^2\ge 2\alpha\log^{2}n$ when $\dbar\chi_n(\zeta)\ne 0$, we obtain
$$|\dbar u_{n,\eta}(\zeta)|^2e^{-nV_n(\zeta)}\le Cn^2|\dbar \chi_n(\zeta)|^2 |\zeta\eta|^{2\lambda-2}e^{-2\alpha\log^{2}n}e^{nV_n(\eta)},$$
whence by \eqref{w-est}
$$\|u_{n,\eta}\|_{nV_n}\le Cn^{1/2}|\zeta_n|^{\lambda-1}e^{-\alpha\log^{2}n}e^{nV_n(\eta)/2}.$$
Applying the pointwise-$L^2$ estimate in Lemma \ref{lemsubm}
(or see \cite[Lemma 4.1]{A2})
$$|u_{n,\eta}(\zeta)|e^{-nV_n(\zeta)/2}\le Cn^{1/2}|\zeta|^{\lambda-1}\|u_{n,\eta}\|_{nV_n}$$
we conclude the proof of the lemma.
\end{proof}

\begin{proof}[Proof of Theorem \ref{asymthm}] Using the above lemmas,
we have for $\zeta\in\calR_n$
\begin{align*}|\bfL_n(\zeta,\zeta)-\bfL_n^*(\zeta,\zeta)|&\le |\bfL_n(\zeta,\zeta)-P_n\bfL_{n,\zeta}^*(\zeta)|+
|u_{n,\zeta}(\zeta)|\\
&\le Cn|\zeta_n|^{2\lambda-2}e^{-\alpha\log^{2}n}\cdot e^{nV_n(\zeta)}.
\end{align*}
By choosing $\alpha>0$ slightly smaller, we obtain that
$$|\bfR_n(\zeta)-n\lambda^2|\zeta|^{2\lambda-2}|\le Ce^{-\alpha\log^{2}n}.$$
The proof is complete. \end{proof}

\end{document}